\documentclass[%
 reprint,
 amsmath,amssymb,
 aps,
pra,
]{revtex4-2}
\usepackage[colorlinks,bookmarksopen,bookmarksnumbered,citecolor=blue,linkcolor=blue,urlcolor=blue]{hyperref}
\usepackage{graphicx}
\usepackage{dcolumn}
\usepackage{bm}

\usepackage{amsthm}

\newtheoremstyle{mythm}
  {6pt}      
  {6pt}      
  {\itshape} 
  {}         
  {\bfseries}
  {.}        
  {0.5em}    
  {}         

\theoremstyle{mythm}
\newtheorem{theorem}{Theorem}
\newtheorem{lemma}{Lemma}
\newtheorem{corollary}{Corollary}

\newcommand{\fangzhen}[1]{{#1}}

\begin{document}

\title{Tabletop reversibility of phase-covariant operations}

\author{Fangzhen Chen}
\author{Xueyuan Hu}%
 \email{xyhu@sdu.edu.cn}
\affiliation{%
 School of Information Science and Engineering, Shandong University, Qingdao 266237, China
}%

\date{\today}

\begin{abstract}
Irreversibility and time-translation symmetry are both fundamental in open quantum dynamics, and it is of great importance to study the interplay between them.
In this paper, we study tabletop reversibility (TTR) for \emph{phase-covariant} quantum operations and obtain sharply different conclusions for finite-dimensional and Gaussian systems.
For finite-dimensional systems,  we prove that any phase-covariant operation that admits a Petz recovery map can be implemented by a dilation which is both time-translational symmetric and tabletop time-reversible.
In contrast, for phase-covariant Gaussian operations we exhibit a concrete obstruction: any Gaussian dilation of a single-mode Gaussian amplification channel is not tabletop time-reversible.
As a byproduct, we find that almost every completely positive and trace preserving (CPTP) map  
admits a dilation that realizes TTR. 
These results clarify whether and how TTR can be realized under phase-covariant symmetry constraints and reveal a qualitative difference between finite-dimensional and Gaussian systems.
\end{abstract}

\maketitle


\section{\label{sec:level1}Introduction}
Irreversibility is a fundamental feature of open quantum dynamics. Once a system interacts with uncontrolled environmental degrees of freedom, information about the input is typically dispersed and cannot be perfectly reconstructed from the output alone. 
In order to retrieve quantum information, the Petz recovery map~\cite{petz1986sufficient,10.1093/qmath/39.1.97,petz2003monotonicity} is proposed from the equality condition of the data processing inequality. It generalizes classical Bayesian retrodiction to quantum settings~\cite{bai2025quantum,surace2022state,cenxin2023quantum} and has important applications in quantum error correction, where it provides a near-optimal operation for recovering quantum information transmitted through noisy channels~\cite{barnum2002reversing,ng2010simple}.

In realistic implementations, every completely positive and trace preserving (CPTP) map admits a quantum dilation \fangzhen{obtained by preparing an auxiliary state, applying a global unitary, and tracing out the auxiliary~\cite{Stinespring1955positive,nielsen2010quantum,watrous2018theory}.} In practice, reversibility is not only a question of whether a mathematical recovery map exists but of whether it can reverse within the same system-environment implementation available on the table. Recently, the notion of tabletop reversibility (TTR) was introduced to describe Petz recovery maps that can be realized by simply inverting the global unitary~\cite{aw2024role}.  
Various specific channels, including product-preserving channels, (generalized) thermal operations, factorizable channels, and controlled-$X$ gates, are shown to be tabletop time reversible~\cite{aw2024role,song2025exact}.
Nevertheless, there are certain dilations of a given channel which are not tabletop time reversible~\cite{aw2024role,song2025exact}.
Remarkably, exact necessary and sufficient conditions for TTR were derived at the level of examining a given unitary dilation~\cite{song2025exact}. 
These works typically focus on a given dilation of a channel and then ask whether the particular \fangzhen{dilation} satisfies TTR.
However, there are different implementations of a given channel.
Therefore, it is natural to ask whether a channel admits a tabletop time-reversible dilation, especially under some physical constraints.

Phase covariance, which arises from \fangzhen{Noether's theorem} and the energy-preserving law~\cite{noether1971invariant,byers1998noether}, is a fundamental symmetry constraint in quantum dynamics.  Such a constraint plays a central role  in quantum clock theory,  where time-translation asymmetry underlies the ability of a quantum system to function as a clock~\cite{marvian2019no}. In quantum thermodynamics, time-translation symmetry also governs the quantum coherence under thermal processing and its coherence generated by asymmetry is an important resource in  thermodynamics~\cite{lostaglio2015quantum,lostaglio2015description}. 
These observations are naturally consistent with the resource theory of asymmetry. In this framework, states and operations that respect the underlying symmetry are regarded as free, whereas those that break the symmetry constitute valuable resources~\cite{chitambar2019quantum}. From this perspective, phase-covariant operations form the class of free operations associated with time-translation symmetry. Correspondingly, a covariant dilation provides a physical implementation of such operations \fangzhen{without introducing any additional asymmetry resources~\cite{marvian2016quantify}}. \fangzhen{Accordingly,} it is natural  to explore TTR of phase-covariant operations within the framework of covariant dilation.

Our aim is to answer the following question: For a phase-covariant operation, does there always exist a covariant dilation that realizes TTR, and if so, can one construct it explicitly? This question is operationally significant. In practice, when one wishes to exploit the TTR property of a channel, the natural objective is not to examine possible dilations one by one but to directly engineer a dilation with the desired reversibility property. Moreover, whenever such a dilation can be explicitly constructed, it is "cost-free" both in asymmetry resource theories and at the level of physical implementation. In particular, both the forward channel and its reverse channel are realized by covariant dilations, so no additional asymmetry resource ‘cost’ in either direction~\cite{marvian2016quantify}. From an operational viewpoint, the significance of TTR shows that the Petz recovery map can be implemented by reusing the same physical devices as in the forward process and inverting unitary, rather than costing extra recovery devices~\cite{song2025exact}.

In this work, we answer this question and uncover sharply different conclusions between finite-dimensional systems and Gaussian dynamics. For finite-dimensional systems, we prove that for any phase-covariant operation with at least one output state that is full rank, there is a covariant dilation that can achieve TTR, and we further construct an explicit dilation that realizes this. In contrast, the Gaussian setting exhibits a genuine obstruction. We identify the single-mode Gaussian amplification channel as a concrete counterexample and show that it fails to admit a Gaussian dilation that realizes TTR for any choice of Gaussian prior and Gaussian reverse bath state. As a by-product, we also show that every finite-dimensional CPTP map whose Petz map is well defined admits a dilation that realizes TTR. Taken together, our work sheds light on an explicit construction of tabletop time-reversible dilation for phase-covariant operations and pinpoints where TTR breaks down in Gaussian dynamics.

\section{Preliminaries}

In this section, we briefly review several concepts.

\textbf{CPTP maps and quantum dilations.}
We consider a system $S$ with Hilbert space $\mathcal {H}_S$ and an auxiliary bath (environment) $B$ with Hilbert space
$ \mathcal{H}_B$. The corresponding Hamiltonians are denoted by $ {H}_S$ and $ {H}_B$. Operators on $ \mathcal H$ are denoted by $\mathcal L(\mathcal{H})$, and density operators are denoted by
$\mathcal S(\mathcal{H})$. A linear map $\mathcal E:\mathcal L(\mathcal H)\to\mathcal L(\mathcal H')$ is called a completely positive and trace preserving map if it is completely positive and satisfies $\operatorname{Tr}[\mathcal E(A)] = \operatorname{Tr}[A]
, \forall\,A\in\mathcal L(\mathcal H).$

Any quantum channel admits a quantum dilation. Concretely, a dilation is constructed by preparing an auxiliary system in the state $\beta$, applying a global unitary $U$ on the joint system and then tracing out the auxiliary~\cite{nielsen2010quantum}. We denote the partial trace over the auxiliary by $\mathrm{Tr}_B$. A quantum dilation of a CPTP map $\mathcal E$ can be presented as
\begin{equation}\label{eq:dilation_forward}
\mathcal E(\rho)=\mathrm{Tr}_B\!\left[U\,(\rho\otimes \beta)\,U^\dagger\right],
\end{equation}
where $U$ is a unitary operator on $\mathcal{H}_S \otimes \mathcal{H}_B$. 
If $\beta$ is pure, the above dilation corresponds to  a Stinespring dilation~\cite{watrous2018theory}.

\textbf{phase-covariant operations.}
Define the time-translation operations generated by the free Hamiltonian $H_S$ as
\begin{equation} \label{time-translation}
    U_t(\cdot):=e^{-iH_S t}(\cdot)e^{iH_S t},\qquad \forall\, t\in\mathbb R .
\end{equation}
A state $\rho\in\mathcal S(\mathcal{H}_S)$ is called phase covariant if $U_t(\rho)=\rho$ for all $t$, or equivalently, $[\rho,H_S]=0$.
A channel $\mathcal E_\mathrm{PC}$ is phase covariant if it commutes with the time-translation operations~\cite{marvian2016quantify},
\begin{equation}\label{eq:phase_covariance}
\mathcal E_\mathrm{PC}\circ U_t = U_t\circ \mathcal E_\mathrm{PC}\qquad \forall\,t\in\mathbb R.
\end{equation}
Such phase-covariant operations are often regarded as free under time-translation symmetry: they preserve the set of covariant states and hence cannot generate coherence between distinct energy eigenspaces from covariant inputs~\cite{marvian2016quantify}. It has been proved that any phase-covariant operation in finite-dimensional systems admits a covariant dilation~\cite{marvian2016quantify}. Namely, for any phase-covariant operation $\mathcal E_{\mathrm{PC}}$, there exist a covariant auxiliary state $\beta$ satisfying $[\beta,H_B]=0$ and a covariant unitary $U$  satisfying $[U,H_S+H_B]=0$, such that $\mathcal E_{\mathrm{PC}}$ can be realized by a dilation of the form in Eq.~(\ref{eq:dilation_forward}).

\textbf{Petz recovery map and TTR.}
For a CPTP map $\mathcal E:\mathcal L(\mathcal{H}_S)\to \mathcal L(\mathcal{H}_S)$, the adjoint $\mathcal E^\dagger$ is defined by~\cite{aw2024role}
\begin{equation}
\mathrm{Tr}\!\left[\mathcal E(A_1)\,A_2\right]=\mathrm{Tr}\!\left[A_1\,\mathcal E^\dagger(A_2)\right],
\qquad \forall\,A_1,A_2\in  \mathcal L(\mathcal{H}_S).
\end{equation}
In the following, we denote the prior state of the system by $\alpha$ and the corresponding output state as $\alpha':=\mathcal E(\alpha)$. 
When $\mathcal E(\alpha)$ is full rank,
the Petz recovery map associated with $(\mathcal E,\alpha)$ is defined as~\cite{petz1986sufficient,10.1093/qmath/39.1.97}
\begin{equation}\label{eq:petz}
\widehat{\mathcal E}_{\alpha}(\cdot)
=\alpha^{1/2}\,\mathcal E^\dagger\!\Big(\mathcal E(\alpha)^{-1/2}\,\cdot\,\mathcal E(\alpha)^{-1/2}\Big)\,\alpha^{1/2}.
\end{equation}
If $\mathcal E(\alpha)$ is not full rank, the operator $\mathcal E(\alpha)^{-1/2}$ is not well defined, and the Petz map admits different definitions~\cite{aw2024role}. To avoid this issue, we assume that $\mathcal E(\alpha)$ is full rank in this paper. 

Fix a quantum dilation of the forward channel as in Eq.~(\ref{eq:dilation_forward}).
For a bath state $\beta'\in\mathcal S( \mathcal{H}_B)$, define the corresponding tabletop reverse channel 
\begin{equation}
\mathcal R_{\beta'}(\cdot):=\mathrm{Tr}_B\!\left[U^\dagger\big(\cdot\otimes\beta'\big)U\right].
\end{equation}

In this paper, we use 
$\beta$ to denote the bath input state in the forward channel and $\beta'$ to denote the reverse bath  state used to implement the tabletop reverse channel.
We say that $\mathcal E$ is tabletop time reversible with respect to the prior $\alpha$ if there exists a bath state $\beta'\in\mathcal S( \mathcal{H}_B)$ such that~\cite{aw2024role}
\begin{equation}\label{eq:ttr}
\widehat{\mathcal E}_{\alpha}=\mathcal R_{\beta'}.
\end{equation}
The key point is that Eq.~(\ref{eq:dilation_forward}) can be interpreted as a fixed device architecture. Equivalently, for a fixed prior $\alpha$, TTR means that the Petz recovery map can be realized within the same device architecture by reversing the global unitary and suitably choosing the bath state $\beta'$.
As an example, thermal operations and generalized thermal maps are known to satisfy TTR~\cite{aw2024role}.

\section{Finite-dimensional systems}

This section investigates the TTR of phase-covariant operations in finite-dimensional systems. The analysis focuses on constructing covariant dilations from the Kraus representation of phase-covariant operations. A central issue in this construction is whether the constraints imposed by TTR can be made compatible with the unitarity condition when transforming the Kraus representation of a channel into its dilation form. It will be shown that such a covariant dilation can indeed be constructed. Furthermore, this construction can be extended to arbitrary finite-dimensional CPTP maps with a full-rank output state for the chosen prior.

Before proceeding to our first main result (Theorem~\ref{theorem}), we introduce the following Lemma, which is essential for the construction of the unitary in the dilation.
\begin{lemma}  \label{lemma1}
Consider an $N$-dimensional complex vector space $\mathbb{C}^N$ with an orthonormal basis $\{|1\rangle,\dots,|N\rangle\}$.
Let $\mathcal J:=\{1,\dots,N\}$ and $\mathcal J_f,\mathcal J_g$ be two subsets of $\mathcal{J}$ satisfying $|\mathcal J_f|=|\mathcal J_g|=m\leq \frac{N}{2}$. The complementary sets of $\mathcal J_f,\mathcal J_g$ are denoted as $\mathcal J_f^c:=\mathcal J\setminus \mathcal J_f$ and
$\mathcal J_g^c:=\mathcal J\setminus \mathcal J_g$.
Suppose $\{|f_j\rangle\}_{j\in\mathcal J_f}$ and $\{|g_i\rangle\}_{i\in\mathcal J_g}$ are two sets of orthonormal vectors in $\mathbb{C}^N$, namely,
\begin{equation}\label{eq:orth}
\begin{aligned}
\langle f_j \mid f_{j'} \rangle &= \delta_{jj'}, \qquad &&\forall\, j,j' \in \mathcal J_f,\\
\langle g_i \mid g_{i'} \rangle &= \delta_{ii'}, \qquad &&\forall\, i,i' \in \mathcal J_g.
\end{aligned}
\end{equation}
Moreover, they satisfy
\begin{equation}\label{zero}
\langle i \mid f_j \rangle = 0,
\quad
\langle g_i \mid j \rangle = 0,
\quad
\forall i\in\mathcal J_g,\; j\in\mathcal J_f.
\end{equation}
If an $N\times N$ complex square matrix $M$ can be written as
\begin{equation}\label{eq:partial-M}
M
=
\sum_{j\in\mathcal J_f} |f_j\rangle\langle j|
+
\sum_{i\in\mathcal J_g} |i\rangle\langle g_i|
+
\sum_{\substack{i\in\mathcal J_g^c\\ j\in\mathcal J_f^c}}
M_{ij}\,|i\rangle\langle j|,
\end{equation}
then the entries $\{M_{ij}\}_{i\in\mathcal J_g^c,\;j\in\mathcal J_f^c}$ can be chosen so that $M$ is unitary.
 
\fangzhen{Specifically, one possible choice of the unspecified entries is
\begin{equation}\label{eq:explicit_Mij}
M_{ij}
=
\sum_{k=2m+1}^{N}
\langle i|u_k\rangle
\langle v_k|j\rangle,
\quad
i\in\mathcal J_g^c,\
j\in\mathcal J_f^c,
\end{equation}
where
$\{|u_k\rangle\}_{k=2m+1}^{N}$ and
$\{|v_k\rangle\}_{k=2m+1}^{N}$ are orthonormal states such that 
$
\{|i\rangle\}_{i\in\mathcal J_g}
\cup
\{|f_j\rangle\}_{j\in\mathcal J_f}
\cup
\{|u_k\rangle\}_{k=2m+1}^{N}
$
and
$
\{|j\rangle\}_{j\in\mathcal J_f}
\cup
\{|g_i\rangle\}_{i\in\mathcal J_g}
\cup
\{|v_k\rangle\}_{k=2m+1}^{N}
$
are both orthonormal bases of $\mathbb{C}^N$.
}
\end{lemma}

\begin{proof}
See proof in Appendix~\ref{app_A}.
\end{proof} 

Lemma~\ref{lemma1} basically means that if $m$ rows and $m$ columns of an $N(\geq2m)$-dimensional square matrix $M$ are specified and satisfy that (C1) the rows and the columns are orthonormal [Eq.~\eqref{eq:orth}] and (C2) the overlap between the specified rows and columns is zero [Eq.~\eqref{zero}], then the unspecified part of $M$ can be chosen such that $M$ is unitary. The proof of Lemma~\ref{lemma1} consists of two steps, which are sketched as follows. We first show that unitary operators $U_1$ and $U_2$ exist such that the specified part of $M$ can be transformed into a unitary $U_s$ acting on a $2m$-dimensional subspace $\mathbb{C}^{2m}$ by the action $U_1(\cdot)U_2^\dagger$. Then we choose the unspecified part of $M'=U_1MU_2^\dagger$ to be block diagonal, with the block acting on $\mathbb{C}^{2m}_c$ being a unitary $U_u$ and the remaining part being zero, where $\mathbb{C}^{2m}_c$ denotes the complement space of $\mathbb{C}^{2m}$. It follows that the unspecified elements of $M$ can be chosen such that $M'=U_s\oplus U_u$, and hence $M=U_1^\dagger(U_s\oplus U_u) U_2$ is unitary.


Now let us consider a system Hamiltonian $H_S$ with the spectral decomposition
\begin{equation}
H_S=\sum_E E\,\Pi_E,
\end{equation} 
where $E$ denotes an energy eigenvalue and $\Pi_E$ is the projector onto the eigensubspace associated with $E$.
For a covariant state $\alpha$ that satisfies $[\alpha,H_S]=0$, there is a common eigenbasis $\{|E,a\rangle\}_{E,a}$ of $H_S$ and $\alpha$, on which both $H_S$ and $\alpha$ are diagonal. The dimension of the system is labeled as $d_S$.

Because $\mathcal E$ is phase covariant, it admits a Kraus decomposition into
Bohr-frequency modes,
\begin{equation} \label{kraus_representation}
\mathcal E(\rho)=\sum_{\omega,s}K_{\omega,s}\rho K_{\omega,s}^\dagger,
\quad
[H_S,K_{\omega,s}]=\omega K_{\omega,s},
\end{equation}
where $\omega$ labels the Bohr-frequency (energy-gap) mode and $s$ labels the multiplicity within each mode.
Equivalently,
\begin{equation}
K_{\omega,s}
=
\sum_{E'-E=\omega}\sum_{a',a}
k^{(\omega,s)}_{E'a',Ea}\,
|E',a'\rangle\langle E,a|.
\end{equation}
Introduce an auxiliary bath Hilbert space $\mathcal H_B$ with orthonormal basis $\{|\mathrm{vac}\rangle,|\omega,s\rangle\}_{\omega,s}$, so that
\begin{equation}
    \mathcal H_B=\mathcal H_{\mathrm{inv}}\oplus \mathcal H^{(\omega,s)},
\end{equation}
where $\mathcal H_{\mathrm{inv}}$ is the one-dimensional subspace spanned by
$|{\mathrm{vac}}\rangle$.
Suppose that the bath Hamiltonian is given by
\begin{equation}
H_B
=
0\,|\mathrm{vac}\rangle\!\langle \mathrm{vac}|
+
\sum_{\omega,s}\omega\,|\omega,s\rangle\!\langle \omega,s|.
\end{equation}

By Theorem~25 of Ref.~\cite{marvian2012symmetry}, any phase-covariant operation $\mathcal E$ admits a covariant Stinespring dilation that can be constructed as follows. Define an isometry $W:\mathcal H_S\otimes \mathcal H_{\mathrm{inv}}\to \mathcal H_S\otimes \mathcal H_B$ as
\begin{equation}\label{eq:mode-decomp}
W=\sum_{\omega,s} K_{\omega,s}\ \otimes\ |{-\omega,s}\rangle\langle{\mathrm{vac}}|.
\end{equation}
Then, the phase-covariant operation $\mathcal E$ can be written as
\begin{equation}\label{eq:Stinespring W}
\mathcal{E}(\rho)=\operatorname{\mathrm{Tr}}_B\!\left[
W\bigl(\rho\otimes |{\mathrm{vac}}\rangle\!\langle{\mathrm{vac}}|\bigr)W^\dagger
\right].
\end{equation}
Notice that each term of $W$ is energy preserving: the operator $K_{\omega,s}$ varies the system energy by $\omega$ while $|{-\omega,s}\rangle\langle{\mathrm{vac}}|$ varies the bath energy by $-\omega$. Therefore, the isometry can be extended to a unitary $U$ acting on $\mathcal H_S\otimes\mathcal H_B$ so that $\mathcal{E}$ can be realized by the following covariant dilation
\begin{equation}\label{eq:Stinespring}
\mathcal{E}(\rho)=\operatorname{\mathrm{Tr}}_B\!\left[
U\bigl(\rho\otimes |{\mathrm{vac}}\rangle\!\langle{\mathrm{vac}}|\bigr)U^\dagger
\right],
\end{equation}
where $[U,\, H_S+H_B]=0$ and $[|{\mathrm{vac}}\rangle\!\langle{\mathrm{vac}}|,\, H_B]=0$.
Our first main result is that for any $\mathcal{E}$, the global unitary $U$ can be constructed such that the covariant dilation in Eq.~\eqref{eq:Stinespring} is tabletop time reversible with respect to $\alpha$ when $\beta'=|\mathrm{vac}\rangle\langle \mathrm{vac}|$, as long as $\mathcal{E}(\alpha)$ is full rank to ensure the existence of the Petz map.


\begin{figure*}
\includegraphics[width=16cm]{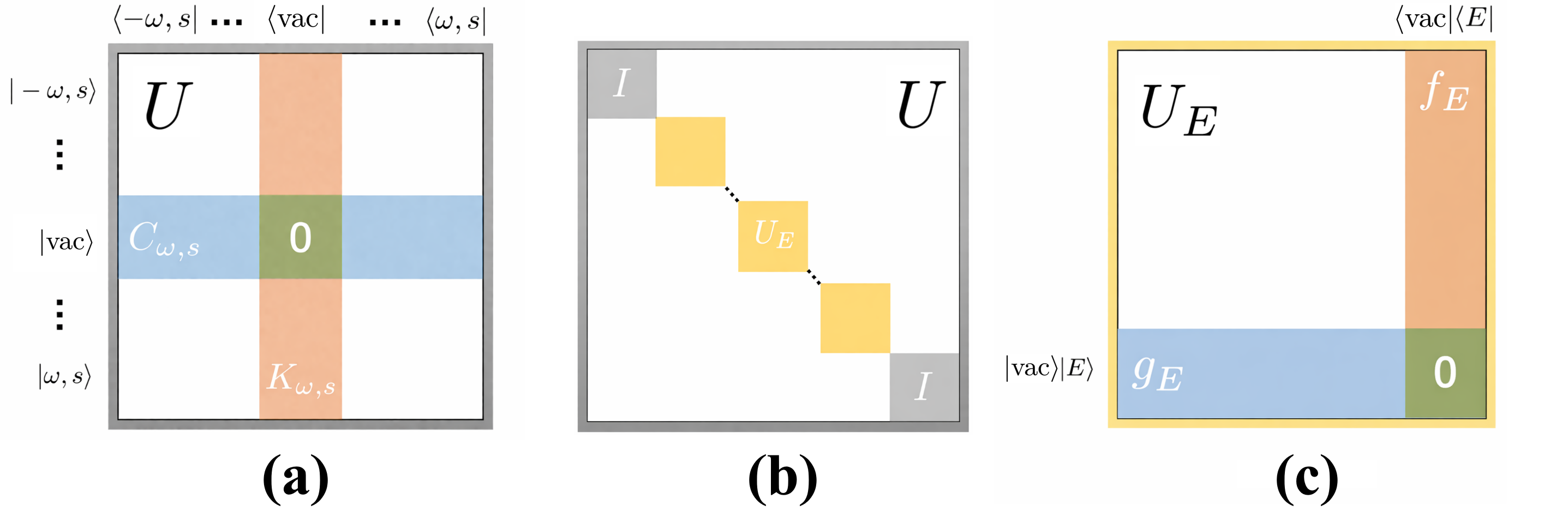}
\caption{\label{TABLETOP/image.png}Schematic illustration of the unitary construction in Theorem~\ref{theorem} and the corresponding proof strategy. For notational convenience, the ordering of labels is reversed in this figure: the first label refers to the environment, while the second refers to the main system. 
(a) Structure of the global unitary \(U\). The \(\langle {\mathrm{vac}}|\)-column is fixed by the forward Kraus operators \(K_{\omega,s}\), the \(|{\mathrm{vac}}\rangle\)-row is fixed by the Kraus operators \(C_{\omega,s}\) of the Petz map, and the overlap block is \(0\). The white blocks remain unspecified at this stage and will be determined by unitary completion.
(b) Block decomposition of the global unitary according to total energy, \(U=\bigoplus_{\bar E} U_{\bar E}\). The yellow blocks represent the constrained energy subblocks \(U_E\) determined by \(K_{\omega,s}\) and \(C_{\omega,s}\). The gray blocks correspond to the parts of $U_{\bar E\neq E}$ and may therefore be chosen as unitary blocks; here, they are taken to be identity operators for convenience.
(c) Structure of a constrained energy subblock \(U_E\). The blue block $g_E$ and the orange block $f_E$ respectively denote the row and column blocks determined by $C_{\omega,s}$ and $K_{\omega,s}$, while their overlap block is fixed to \(0\). By Lemma~\ref{lemma1}, each such subblock $U_E$ can be completed to a unitary. Combined with (b), this implies that the entire global unitary \(U\) can be completed to a unitary.}
\end{figure*}

\begin{theorem} \label{theorem}
    For any phase-covariant operation $\mathcal E$, if a covariant state $\alpha$ exists such that $\alpha'\equiv \mathcal{E}(\alpha)$ is full rank, there is a dilation of $\mathcal E$ that is both covariant and tabletop time reversible with respect to the prior $\alpha$.
\end{theorem}


\fangzhen{Before providing the complete proof of Theorem~\ref{theorem}, we briefly outline the main idea of the proof. As illustrated in Fig.~\ref{TABLETOP/image.png}, our purpose is to show that the Kraus operators of the forward
channel and those of its Petz recovery map can be embedded into a single energy-preserving unitary $U$. As shown in Eq.~\eqref{eq:mode-decomp}, the Kraus operators
\(K_{\omega,s}\) of the forward channel determine the isometry $W$ and thus $d_S$ columns of $U$. Similarly, the Kraus operators \(C^\dagger_{\omega,s}\) of the Petz map determine $d_S$ columns of $U^\dagger$. In other words, $d_S$ columns and $d_S$ rows of $U$ are determined by the forward channel and its Petz map respectively. Importantly, in our setting, the input and output states of the auxiliary are orthogonal to each other. It ensures that the overlap between the prescribed rows and columns is zero. Therefore, we can apply Lemma~\ref{lemma1} to show that such unitary $U$ does exist. Moreover, the covariance of the forward and backward channels also ensures the block diagonal form of $U$, and thus $U$ is also energy preserving.
}

\begin{proof}

The Petz recovery map for the prior $\alpha$ is calculated as
\begin{equation}\label{reverse_map_kraus}
\begin{aligned}
\widehat{\mathcal E}_{\alpha}(\cdot)
&=
\alpha^{1/2}\mathcal E^\dagger\!\left[\alpha'^{-1/2}(\cdot)\alpha'^{-1/2}\right]\alpha^{1/2} \\
&=
\sum_{\omega,s}
\alpha^{1/2}K_{\omega,s}^\dagger\alpha'^{-1/2}
\,(\cdot)\,
\alpha'^{-1/2}K_{\omega,s}\alpha^{1/2} \\
&=
\sum_{\omega,s} C_{\omega,s}^\dagger (\cdot)\, C_{\omega,s}.
\end{aligned}
\end{equation}
Furthermore,
\begin{equation}\label{eq:C-unital}
\begin{aligned}
\sum_{\omega,s} C_{\omega,s} C_{\omega,s}^\dagger
&=
\alpha'^{-1/2}
\Bigl(\sum_{\omega,s} K_{\omega,s}\alpha K_{\omega,s}^\dagger\Bigr)
\alpha'^{-1/2} \\
&=
\alpha'^{-1/2} \alpha' \alpha'^{-1/2} =
I.
\end{aligned}
\end{equation}
The above equation always holds since the Petz map is a CPTP map  \fangzhen{while $\alpha'$ is full-rank}.
Thus, $\{C_{\omega,s}^\dagger\}_{\omega,s}$ forms a valid Kraus representation of the Petz map.

Now consider a tabletop reverse channel given by
\begin{equation} \label{reverse channle}
\mathcal R(\cdot)=\operatorname{Tr}_B\!\left[U^\dagger \bigl(\cdot\otimes |{\mathrm{vac}}\rangle\langle {\mathrm{vac}}|\bigr) U\right],
\end{equation}
where $U$ is the unitary in the dilation of the forward channel $\mathcal{E}$ as in Eq.~\eqref{eq:Stinespring}.
Hence a column of blocks in $U$ is fixed as
\begin{eqnarray}
\label{unitary construction column} \langle -\omega,s|\,U\,|{\mathrm{vac}}\rangle =& K_{\omega,s},\\
\label{unitary construction block} \langle {\mathrm{vac}}|\,U\,|{\mathrm{vac}}\rangle =& 0.
\end{eqnarray}
In the following, we show that $U$ can be constructed such that (a) the tabletop reverse channel $\mathcal R$ in Eq.~\eqref{reverse channle} coincides with the Petz map $\widehat{\mathcal E}_{\alpha}$ defined in Eq.~\eqref{reverse_map_kraus}, and (b) $[U,H_S+H_B]=0$.

Condition (a) is equivalent to
\begin{equation}\label{unitary construction row}
\langle {\mathrm{vac}}|\,U\,|\omega,s\rangle = C_{\omega,s}.
\end{equation}
It means that a row of blocks in $U$ is also specified. Moreover, the overlap block between the specified row and column equals zero.

For condition (b), we write the total Hamiltonian as $H_S+H_B=\sum_{\bar{E}}\bar{E}\Pi_{\bar{E}}$, where $\Pi_{\bar{E}}$ is the projection onto the eigensubspace $\mathcal{H}_{\bar{E}}$ belonging to $\bar{E}$. With this, condition (b) is equivalent to
\begin{equation}\label{eq: U block}
U = \bigoplus_{\bar{E}} U_{\bar{E}},
\end{equation}
where $U_{\bar{E}}$ is a unitary that acts on $\mathcal{H}_{\bar{E}}$.

Next, we show that $U$ exists such that \fangzhen{Eqs.~\eqref{unitary construction column}--~\eqref{eq: U block}} can be satisfied.

Firstly, all of the matrix elements specified in \fangzhen{Eqs.~\eqref{unitary construction column}--~\eqref{unitary construction row}} lie in the blocks $U_{\bar{E}\in \{E\}}$, where $\{E\}$ denotes the set of eigenvalues of $H_S$. Therefore, each of the blocks $U_{\bar{E}\notin \{E\}}$ can be chosen as an arbitrary unitary matrix. Here we \fangzhen{choose} $U_{\bar{E}\notin \{E\}}=\mathbb{I}$ for simplicity.

Secondly, all nonzero elements of $U$ specified in Eqs.~\eqref{unitary construction column} and~\eqref{unitary construction row} lie in the blocks $U_{\bar{E}=E}$, where $E$ denotes the eigenvalues of $H_S$. The reason is as follows. The elements specified by Eq.~\eqref{unitary construction column} read
\begin{equation}
    \langle-\omega,s|\langle E',a|U|E,a\rangle|\mathrm{vac}\rangle=\langle E',a|K_{\omega,s}|E,a\rangle,
\end{equation}
which is nonzero only if $E'-E=\omega$. 
It follows that for any nonzero element of $U$, both the column index $|E,a\rangle|\mathrm{vac}\rangle$ and the row index $|E',a\rangle|-\omega,s\rangle$ belong to $\mathcal{H}_E$, and therefore these elements are located within the blocks $U_E$. 
A similar argument applies to the elements specified by Eq.~\eqref{unitary construction row}.
Therefore, we are free to set all the elements that do not lie in the blocks $U_{\bar E}$ to be zero, such that $U$ is block diagonal.

Thirdly, the elements specified by Eqs.~\eqref{unitary construction column} and~\eqref{unitary construction row} in each block $U_E$ satisfy the conditions of Lemma~\ref{lemma1}. More precisely, in each block $U_E$, the columns specified by Eq.~\eqref{unitary construction column} are
\begin{equation}
\begin{aligned}
|f_{E,a}\rangle
&:= U|E,a\rangle|\mathrm{vac}\rangle = W|E,a\rangle|\mathrm{vac}\rangle \\
&= \sum_{\omega,s} K_{\omega,s}|E,a\rangle|-\omega,s\rangle .
\end{aligned}
\end{equation}
Similarly, the rows specified by Eq.~\eqref{unitary construction row} are
\begin{equation}
\langle g_{E,a}| = \langle \mathrm{vac}|\langle E,a|\,U = \sum_{\omega,s} \langle E,a|\, C_{\omega,s}\otimes \langle \omega,s|.
\end{equation}
Clearly, $(\langle \mathrm{vac}|\langle E,a'|)\,|f_{E,a}\rangle = 0,\
\langle g_{E,a}|\, (|E,a'\rangle|\mathrm{vac}\rangle)=0,$ for all $a,a'$. Moreover, both the specified columns and the specified rows are orthonormal, since
\begin{equation} \label{orth_f}
\langle f_{E,a}|f_{E',b}\rangle
=
\langle E,a|
\Bigl(\sum_{\omega,s}K_{\omega,s}^\dagger K_{\omega,s}\Bigr)
|E',b\rangle
=
\delta_{EE'}\delta_{ab}
\end{equation}
and
\begin{equation}\label{orth_g}
\langle g_{E,a}|g_{E',b}\rangle
=
\langle E,a|
\Bigl(\sum_{\omega,s}C_{\omega,s}C_{\omega,s}^\dagger\Bigr)
|E',b\rangle
=
\delta_{EE'}\delta_{ab}.
\end{equation}
By Lemma~\ref{lemma1}, each block $U_E$ can be extended to a unitary from the specified rows and columns.

In summary, for every phase-covariant operation $\mathcal{E}$, if its Petz map is well defined with respect to a covariant prior state $\alpha$, we can construct a dilation that is both phase covariant and tabletop time reversible with respect to $\alpha$. This completes the proof of Theorem~\ref{theorem}.
\end{proof}

Our construction of the dilation relies on two key ingredients.
First, when the prior is chosen to be covariant, the Petz recovery map is a phase-covariant operation, and its Kraus operators can fit into a covariant unitary.
Second, in both the forward and tabletop reverse channels, the auxiliary input and output states belong to mutually orthogonal subspaces. Precisely, we choose the input state of the auxiliary as $|\mathrm{vac}\rangle\langle \mathrm{vac}|$ and let the corresponding block of $U$ vanish, i.e., $\langle \mathrm{vac}|U|\mathrm{vac}\rangle=0$.
With these two key points, we show that the Kraus operators of the forward channel can fit into a column of blocks of a global covariant unitary, while the Hermitian conjugates of Kraus operators of the Petz recovery map can fit into a row of blocks of the same unitary. 
It follows that the constructed tabletop reverse channel equals the Petz map.


Moreover, the dilation provided by Theorem~\ref{theorem} is cost free at both the resource-theoretic and device levels. Indeed, both the forward channel and the corresponding reverse channel are realized as phase-covariant dilations, so no additional asymmetry resource is invoked. For physical implementation, TTR implies that the Petz recovery map can be \fangzhen{realized} by reversing the setup of the forward process instead of requiring additional devices.

Theorem~\ref{theorem} identifies a different class of dilations that satisfy TTR.
It has been shown that~\cite{aw2024role}, if a tuple $(U,\alpha,\beta)$ is product preserving, i.e., $U(\alpha\otimes\beta)U^\dagger=\alpha'\otimes\beta'$, then the dilation $\mathcal E(\cdot)=\mathrm{Tr}_B[U(\cdot\otimes\beta)U^\dagger]$ is tabletop time reversible with respect to $\alpha$. That is, $\widehat{\mathcal E}_\alpha(\cdot)=\mathcal{R}_{\beta'}(\cdot)=\mathrm{Tr}_B[U^\dagger(\cdot\otimes\beta')U]$.
However, in the dilation constructed in Theorem~\ref{theorem}, the tuple $(U,\alpha,|\mathrm{vac}\rangle\langle \mathrm{vac}|)$ is not product preserving. 
Indeed,
\begin{equation}
\begin{aligned}
U(\alpha\otimes |\mathrm{vac}\rangle\langle \mathrm{vac}|)U^\dagger
=\\
\sum_{\omega',s'}\sum_{\omega,s}
K_{\omega',s'}\alpha K_{\omega,s}&^\dagger 
 \otimes
|-\omega',s'\rangle\langle -\omega,s|,
\end{aligned}
\end{equation}
which is generally not a product state and contains nontrivial system-bath correlations. 

Our results also have applications in the resource theory of quantum thermodynamics. In this resource theory, the free state is the Gibbs state $\gamma_S^\kappa := \frac{e^{-\kappa H_S}}{\mathrm{Tr}(e^{-\kappa H_S})}$, where $\kappa := \frac{1}{k_B T}>0$ is the inverse temperature, $k_B$ is the Boltzmann constant and $T$ is the temperature. There are two widely accepted sets of free operations. The first set is called  thermal operations (TOs)~\cite{horodecki2013fundamental}. A channel $\mathcal E$ is a  TO)  at inverse temperature $\kappa$ if it can be implemented as $\mathcal E(\rho)=\mathrm{Tr}_B\!\left[\,U\big(\rho\otimes\gamma_B^\kappa\big)U^\dagger\right]$, where $[U,\,H_S+H_B]=0$ and $\gamma_B^\kappa := \frac{e^{-\kappa H_B}}{\mathrm{Tr}(e^{-\kappa H_B})}$ is the Gibbs state of the auxiliary (heat bath).
The second set is called enhanced thermal operations (EnTOs)~\cite{cwiklinski2014towards,lostaglio2015quantum}. This set is defined axiomatically as channels that satisfy (a) covariance with respect to $H_S$ and (b) Gibbs state preservation.
It is known that~\cite{aw2024role} all TOs are tabletop time-reversible with respect to the Gibbs state.
Yet, the set of TOs is a strict subset of EnTOs~\cite{ding2021exploring}.
It then leaves open the question of TTR for EnTOs.
Here, we answer this question in the affirmative.

\begin{corollary} \label{corollary_1}
    Every EnTO admits a covariant dilation that is tabletop time reversible with respect to the Gibbs state $\gamma_S^\kappa$.
\end{corollary}

\begin{proof}
    EnTOs preserve the Gibbs state $\gamma_S^\kappa$. Choosing the prior to be $\alpha=\gamma_S^\kappa$, we have the corresponding output state $\alpha'=\mathcal E(\alpha)=\gamma_S^\kappa$. It follows that $\alpha'$ is full rank. In addition, EnTOs are phase covariant. According to Theorem~\ref{theorem}, for any EnTO, there is a dilation which is both phase covariant and tabletop time reversible. \fangzhen{This completes the proof of Corollary~\ref{corollary_1}.}
\end{proof}

It is worth noting that the implementation of EnTOs is not thermodynamically cost free: there are EnTOs that cannot be \fangzhen{realized} by preparing an auxiliary Gibbs state and applying a phase-covariant unitary. However, our result shows that with some athermality resource in the auxiliary state, one can \fangzhen{realize} an EnTO by using a phase-covariant unitary. Moreover, this dilation is tabletop time reversible, and thus the reverse channel does not cost extra devices.

\fangzhen{
The athermality resource mentioned above includes two contributions: the coherence between energy levels and the non-equilibrium in population distribution.
Apparently, the coherence resource vanishes in the above EnTO implementation. In the dilation considered in Corollary~\ref{corollary_1}, both the auxiliary state and the global unitary are covariant. Thus the implementation does not consume any coherence resource.
However, the initial state $|\mathrm{vac}\rangle$ of the auxiliary is not a Gibbs state, so some thermodynamic resource, such as energy, has to be consumed to prepare the initial state. According to Refs.~\cite{brandao2013resource,wilming2017axiomatic}, the minimum consumed energy in preparing a state $\beta_B$ which satisfies $[\beta_B,H_B]=0$ is quantified by the nonequilibrium free energy cost
\begin{equation} \label{deltaF}
    \Delta F_\kappa(\beta_B)
    :=
    F_\kappa(\beta_B)-F_\kappa(\gamma_B^\kappa)
    =
    \kappa^{-1} D(\beta_B\Vert \gamma_B^\kappa),
\end{equation}
where \(\kappa\) denotes the inverse temperature, 
$\gamma_B^\kappa=e^{-\kappa H_B}/Z_B$ is the Gibbs state,
\(F_\kappa(\rho):=\operatorname{Tr}[H\rho]-\kappa^{-1}S(\rho)\) is the free energy associated with the Hamiltonian \(H\), 
\(S(\rho):=-\operatorname{Tr}[\rho\ln\rho]\) is the von Neumann entropy, and 
\(D(\rho\Vert\sigma)\) denotes the quantum relative entropy.
$
D(\rho\Vert\sigma)
:=
{\rm Tr}\!\left[\rho(\log\rho-\log\sigma)\right].
$
In our case, $\beta_B=|\mathrm{vac}\rangle\langle \mathrm{vac}|$, and direct calculation leads to
\begin{align} \label{Cath}
  \Delta F_\kappa(|\mathrm{vac}\rangle\langle \mathrm{vac}|)=
    \kappa^{-1}\log Z_B.
\end{align}}

Next, we present a concrete example to show how to construct a covariant dilation that is tabletop time reversible for an EnTO that does not belong to TOs.


\textbf{A qutrit example.}
Consider a qutrit with Hamiltonian $H_S=\sum_{E=0}^2 E\Delta\,|E\rangle_S\langle E|$. 
The Gibbs state is $\gamma_S^\kappa=\frac{1}{1+q+q^2}
\bigl(
|0\rangle_S\langle 0|
+q\,|1\rangle_S\langle 1|
+q^2\,|2\rangle_S\langle 2|
\bigr)$, where $q=e^{-\kappa \Delta}$.
As shown in Ref.~\cite{ding2021exploring}, the channel characterized by the Kraus operators $K_{+2}=q\,|2\rangle_S\langle 0|$,
$K_{0}=\sqrt{1-q^2}\,|0\rangle_S\langle 0|+|1\rangle_S\langle 1|$,
$K_{-2}=|0\rangle_S\langle 2|$ and \(K_{\pm1}=0\) is an EnTO, but it can induce state transformations that cannot be realized via TOs even approximately.
In the following, we construct a covariant dilation of this channel, which is tabletop time-reversible with respect to the Gibbs state $\gamma_S^\kappa$.

The auxiliary is a six-dimensional system with Hamiltonian $H_B=\sum_{\omega=0,\pm1,\pm2}\omega\Delta\,|\omega\rangle_B\langle\omega|
+0\,|\mathrm{vac}\rangle_B\langle \mathrm{vac}|$ and is initially in state $|\mathrm{vac}\rangle\langle \mathrm{vac}|$. The global covariant unitary is block diagonal as in Eq.~\eqref{eq: U block}. If $\bar{E}\neq0,\Delta,2\Delta$, $U_{\bar{E}}$ can be chosen as arbitrary unitary matrices. The three blocks with $\bar{E}=0,\Delta,2\Delta$ are in the following form
\begin{align}\label{qutrit_example_result}
U_{E_1=0}
&=
\begin{pmatrix}
1-q^2 & 0 & -q\sqrt{1-q^2} & q\\
0 & 1 & 0 & 0\\
-q\sqrt{1-q^2} & 0 & q^2 & \sqrt{1-q^2}\\
q & 0 & \sqrt{1-q^2} & 0
\end{pmatrix}, \\
U_{E_2=\Delta}
&=
\begin{pmatrix}
1 & 0 & 0 & 0\\
0 & 0 & 1 & 0\\
0 & 1 & 0 & 0\\
0 & 0 & 0 & 1
\end{pmatrix},
U_{E_3=2\Delta}=
\begin{pmatrix}
1 & 0 & 0 & 0\\
0 & 0 & 0 & 1\\
0 & 0 & 1 & 0\\
0 & 1 & 0 & 0 
\end{pmatrix}.\notag
\label{eq:U_blocks}
\end{align}
Here, the corresponding eigenspaces are $\mathcal H_{E_1=0}=\operatorname{span}\{|2\rangle_S|-2\rangle_B,\ |1\rangle_S|-1\rangle_B,\ |0\rangle_S|0\rangle_B,\ |0\rangle_S|\mathrm{vac}\rangle_B\}$, $\mathcal H_{E_2=\Delta}=\operatorname{span}\{|2\rangle_S|-1\rangle_B,\ |1\rangle_S|0\rangle_B,\ |1\rangle_S|\mathrm{vac}\rangle_B,\ |0\rangle_S|+1\rangle_B\}$, and
$\mathcal H_{E_3=2\Delta}=\operatorname{span}\{|2\rangle_S|0\rangle_B,\ |2\rangle_S|\mathrm{vac}\rangle_B,\ |1\rangle_S|+1\rangle_B,\ |0\rangle_S|+2\rangle_B\}$. With this, one can show that the tabletop reverse channel is characterized by its Kraus operators $C_\omega^\dagger = K_{-\omega}$, which coincides with the Petz recovery map $\widehat{\mathcal E}_{\gamma_S^\kappa}$.
\fangzhen{Details of constructing the unitary blocks in Eq.~\eqref{qutrit_example_result} can be found in Appendix~\ref{Worked Example}}.

\fangzhen{
Next, we briefly evaluate the athermal resource consumed in the dilation. With $H_B=\sum_{\omega=0,\pm1,\pm2}\omega\Delta\,|\omega\rangle_B\langle\omega|
+0\,|\mathrm{vac}\rangle_B\langle \mathrm{vac}|$, it is easy to check that
    $Z_B
    =
    {\rm Tr}\,e^{-\kappa H_B}=
    q^{-2}+q^{-1}+2+q+q^2$.  
    Then by Eq.~\eqref{Cath}, the corresponding nonequilibrium free energy cost is
\begin{equation}
   \Delta F_\kappa(\beta_B)
    =
    \kappa^{-1}
    \log\left(q^{-2}+q^{-1}+2+q+q^2\right).
    \label{eq:qutrit_athermality_work}
\end{equation}
It means that by consuming this amount of athermality resource, the EnTO channel can be realized by a tabletop time-reversible and covariant dilation.

}




Another interesting application of Theorem~\ref{theorem} is the following corollary, which states that almost every CPTP map has a dilation that is tabletop time reversible.
\begin{corollary}  \label{corollary_2}
    For any CPTP map $\mathcal E$, if a state $\alpha$ exists such that $\alpha'\equiv \mathcal{E}(\alpha)$ is full rank, there is a dilation of $\mathcal E$ that is tabletop time reversible with respect to the prior $\alpha$. 
\end{corollary}
\begin{proof}
When the system Hamiltonian is trivial \(H_S = 0\), the time-translation operations in Eq.~\eqref{time-translation} reduce to the identity map, namely, $U_t(\rho)=\rho,\forall\, t\in\mathbb R,\rho$. Hence the covariance condition in Eq.~\eqref{eq:phase_covariance} is automatically satisfied by every CPTP map. In this case, the proof of Theorem~\ref{theorem} applies with only minor modifications:
(a) the system basis is labeled solely by the multiplicity index \(\{|a\rangle\}_a\); (b) the only Bohr frequency is \(\omega=0\), therefore, the Kraus operators are indexed solely by \(s\), and the auxiliary bath basis can be written as \(\{|\mathrm{vac}\rangle, |s\rangle\}_s\); (c) the condition \([\alpha,H_S]=0\) becomes trivial, so any prior \(\alpha\) with full-rank output \(\alpha'=\mathcal E(\alpha)\) is admissible; and (d) the unitary in Eq.~\eqref{eq: U block} only consists of a single block \(U_{E}\). This completes the proof of Corollary~\ref{corollary_2}.
\end{proof}

Corollary~\ref{corollary_2} provides a constructive paradigm for exploiting TTR. Whenever one wishes to use the Petz recovery map together with the forward channel to reconstruct the input information, such a construction may offer practical convenience. In particular, the Petz recovery map can be implemented using the same device as the forward channel, with the only changes being to prepare the same bath state and invert the original unitary. This is particularly meaningful because the Petz map already plays a central role in quantum error correction: it is known to be near optimal for approximately correctable codes~\cite{barnum2002reversing,ng2010simple}.

\section{Gaussian systems}

In this section, we study the TTR for Gaussian covariant operations. We first briefly review the characterization of Gaussian states and operations, and then identify a class of Gaussian covariant operations \fangzhen{for} which no Gaussian dilation is tabletop time reversible. This is in contrast with our result for discrete-variable systems and reveals the constraints on TTR imposed by Gaussianity.

\subsection{Gaussian framework and notation}

For an $n$-mode bosonic system, let $\hat{\bm r}:=(\hat x_1,\hat p_1,\dots,\hat x_n,\hat p_n)^T$ be the vector of quadrature operators, satisfying the commutation relations 
\begin{equation}
    [\hat r_k,\hat r_l]=i\Omega_{kl},\qquad
    \Omega=\bigoplus_{k=1}^n
\begin{pmatrix}
0&1\\
-1&0
\end{pmatrix}.
\end{equation}
A Gaussian state is completely characterized by its displacement vector $\vec d$ and covariance matrix $V$, defined by
\begin{equation}
d_k:=\langle \hat r_k\rangle,
\qquad
V_{kl}:=\frac{1}{2}\langle \{\hat r_k-d_k,\hat r_l-d_l\}\rangle,
\end{equation}
where $\{\cdot,\cdot\}$ denotes the anticommutator.
The covariance matrix satisfies $V=V^T$ and $V+i\Omega\succeq 0$.
A Gaussian state $\rho$ is strictly positive if and only if its covariance matrix satisfies $V+i\Omega\succ 0$; such a state is called a faithful state.


A Gaussian channel transforms covariance matrices and displacement vectors as~\cite{weedbrook2012gaussian}
\begin{equation}
 V \longmapsto XVX^T+Y,\qquad \vec d\longmapsto X\vec d+\vec \delta,
\end{equation}
where $\vec \delta$ is a real vector, $X$ is a real matrix, and $Y$ is a symmetric matrix. The matrices $X$ and $Y$ further satisfy 
\begin{equation} \label{positive_condition}
Y+i\Omega-iX\Omega X^T\succeq 0.
\end{equation}
In particular, when a Gaussian channel is a Gaussian unitary, it follows that $X=S$ and $Y=0$~\cite{weedbrook2012gaussian,gianfelici2021hierarchy}, where $S$ is a symplectic matrix satisfying $S\Omega S^T=\Omega$.

\textbf{Gaussian Petz recovery map:}
Let $\alpha$ be an $n$-mode Gaussian prior with displacement vector ${\vec d_\alpha}$ and covariance matrix $V_\alpha$.
After the action of a Gaussian operation characterized by $(X,Y)$, the covariance matrix and displacement vector of the corresponding output state $\alpha'$ are
\begin{equation}\label{V_alpha}
V_{\alpha'} := X V_\alpha X^{T}+Y,\qquad {\vec d_{\alpha'}} := X {\vec d_\alpha}+\vec \delta.
\end{equation}
Assume that $\alpha'$  is a full-rank $n$-mode Gaussian state with displacement vector $\vec d_{\alpha'}$ and covariance matrix
$V_{\alpha'}$.
Then the Gaussian Petz recovery map
$\widehat{\mathcal E}_{\alpha}$ acts on covariance matrices and displacement vectors as follows~\cite{lami2018approximate}:
\begin{equation}\label{eq:gaussian_petz_recovery}
\widehat{\mathcal E}_{\alpha}:\quad
V \longmapsto X_{\mathrm P} V X_{\mathrm P}^{T}+Y_{\mathrm P},\qquad
\vec d \longmapsto X_{\mathrm P} \vec d+\vec \delta_{\mathrm P}.
\end{equation}
Here
\begin{align}
X_{\mathrm P}
&:= \sqrt{I+(V_\alpha\Omega)^{-2}} V_\alpha X^{T}
\Big(\sqrt{I+(\Omega V_{\alpha'})^{-2}}\Big)^{-1} V_{\alpha'}^{-1},
\label{XP}
\\
Y_{\mathrm P}
&:= V_\alpha - X_{\mathrm P} V_{\alpha'} X_{\mathrm P}^{T},
\label{YP}
\\
\vec \delta_{\mathrm P}
&:= {\vec d_\alpha} - X_{\mathrm P} {\vec d_{\alpha'}}
= {\vec d_\alpha} - X_{\mathrm P}(X {\vec d_\alpha}+\vec \delta).
\label{deltaP}
\end{align}

\subsection{No TTR of single-mode Gaussian amplification channel}
Amplification channels are Gaussian phase-covariant channels~\cite{gianfelici2021hierarchy}.
A single-mode Gaussian amplification channel is characterized as
\begin{equation} \label{Gaussian_amplification}
\mathcal{E}_{A}:\qquad
V \longmapsto \tau V + (\tau-1)(2\bar n+1) I,
\qquad
\tau>1 .
\end{equation}
It is typically implemented by preparing a single-mode auxiliary in a Gibbs state with mean excitation number $\bar n$ and applying a two-mode squeezing operation~\cite{koukoulekidis2025symmetry,weedbrook2012gaussian}.
This is a Gaussian dilation of the amplification channel. 
Note also that this dilation is not covariant with respect to the free Hamiltonian $H_S+H_B$, because the two-mode squeezing does not preserve the total energy.
In the following, we will show that any Gaussian (but not necessarily covariant) dilation of an amplification channel is not tabletop time reversible.

\begin{theorem} \label{theorem_2}
Let $\mathcal{E}_{A}$ be the single-mode Gaussian amplification channel. There exists no Gaussian dilation of $\mathcal{E}_{A}$ whose tabletop reverse channel
coincides with the Petz recovery map with respect to any faithful single-mode Gaussian prior $\alpha$.
Consequently, not every Gaussian phase-covariant operation admits a Gaussian dilation that realizes TTR.
\end{theorem} 

\begin{proof}
Consider an arbitrary Gaussian dilation of a single-mode amplification channel $\mathcal{E}_{A}$. The auxiliary is an $n_B$ mode system with covariance matrix $V_B$. The symplectic matrix corresponding to the global Gaussian unitary can be written in block form as
\begin{equation}
S=\begin{pmatrix}
S_{11} & S_{12}\\
S_{21} & S_{22}
\end{pmatrix},
\end{equation}
where $S_{11}$ is a $2\times2$ matrix and $S_{22}$ is a $2n_B\times 2n_B$ matrix.
From $S^{T}\Omega S=\Omega$, one has $S^{-1}=\Omega^{-1}S^{T}\Omega$, and it follows that
\begin{equation}\label{eq: backward U}
S^{-1}
=
\begin{pmatrix}
\Omega^{-1} S_{11}^{T} \Omega
&
\Omega^{-1} S_{21}^{T} \Omega_{B}
\\[2mm]
\Omega_{B}^{-1} S_{12}^{T} \Omega
&
\Omega_{B}^{-1} S_{22}^{T} \Omega_{B}
\end{pmatrix}.
\end{equation}
The forward channel acts on the covariance matrices as
\begin{equation} \label{Gaussian_forward}
\left.V \longmapsto S(V\oplus V_{B})S^{T}\right|_{S}=S_{11} V S_{11}^T + S_{12} V_B S_{12}^T,
\end{equation}
where $\cdot|_S$ denotes taking the top-left $2\times 2$ block corresponding to the system $S$.

Because the channel is an amplification channel as in Eq.~\eqref{Gaussian_amplification}, by comparing Eqs.~\eqref{Gaussian_forward} and~\eqref{Gaussian_amplification}, we have
\begin{equation}\label{eq: forward X}
S_{11} = X_\mathrm{for} = \sqrt{\tau}\, I .
\end{equation}
Now consider the tabletop reverse channel generated by the Gaussian unitary $U^\dagger$,
\begin{equation}\label{eq:reverse}
V \longmapsto \left. S^{-1}(V\oplus V_{B}^{\prime}){S^{T}}^{-1}\right|_{S}=X_\mathrm{rev}VX_\mathrm{rev}^T+Y_\mathrm{rev}.
\end{equation}
Using Eq.~\eqref{eq: backward U}, we have
\begin{equation} \label{X_rev_amplification}
X_{\mathrm{rev}}
=
\Omega^{-1}S_{11}^T\Omega
=
\sqrt{\tau}\,I .
\end{equation}
Hence, the determinant of $X_{\mathrm{rev}}$ corresponding to the reverse channel is
\begin{equation}
\det X_{\mathrm{rev}}=\tau>1.
\end{equation}

Next, we calculate the Petz recovery map of $\mathcal{E}_A$.
For a $2\times 2$ real symmetric matrix $V$, it is easy to check that $(V\Omega)^2=-\det(V)\,I$.
Substituting this and Eq.~\eqref{eq: forward X} into Eq.~\eqref{XP}, we obtain
\begin{equation} \label{eq:XP_single_mode}
X_\mathrm{P}
=
\sqrt{\tau}\,
\frac{\sqrt{1-\nu^{-1}}}{\sqrt{1-\nu'^{-1}}}\,
V_\alpha V_{\alpha'}^{-1},
\end{equation}
where $\nu\equiv\det V_\alpha>1$ and $\nu'\equiv\det V_{\alpha'}>1$.
It follows that
\begin{equation}\label{detXP}
\det X_\mathrm{P}
=
\tau\,\frac{\nu-1}{\nu'-1}.
\end{equation}
From $V_{\alpha'}=\tau V_\alpha+ (\tau-1)(2\bar n+1)I$, we have $\nu'>\tau^2\nu>\nu$. A direct calculation leads to $\det X_\mathrm{P}<1$. Therefore, for an amplification map, the Petz recovery map with respect to any faithful prior does not equal the tabletop reverse channel corresponding to any Gaussian dilation. This completes the proof of Theorem~\ref{theorem_2}.
\end{proof}

 
Intuitively, a Gaussian amplification channel amplifies the energy of an input state regardless of its initial energy, so the channel that is designed to recover the initial state should decrease the energy, and the Petz recovery map must be a loss channel.
However, as we have shown in the proof, the tabletop reverse channel for any dilation of an amplification channel would also amplify the energy for any input state. 
Therefore, the Petz recovery map cannot coincide with the tabletop
reverse channel of any Gaussian dilation of $\mathcal{E}_{A}$.

The conclusion in the Gaussian setting is qualitatively opposite to the finite-dimensional one. \fangzhen{In Appendix~\ref{app:asymptotic_nonGaussian_TTR_amp}, we also construct a phase-covariant but non-Gaussian dilation that realizes TTR.} This reflects the fact that Gaussianity
imposes highly restrictive structural constraints on dilations and recovery processes. Moreover, such
constraints can obstruct the implementability of TTR.
It is consistent with a broader theme in continuous-variable quantum information: Gaussianity often imposes strong structural constraints and leads to no-go results. Representative examples include the impossibility of Gaussian entanglement distillation under Gaussian operations~\cite{eisert2002distilling,fiuravsek2002gaussian} and the no-go theorem for Gaussian quantum error correction~\cite{niset2009no}. From this perspective, our result adds TTR to the list of physical properties that can be fundamentally restricted by the Gaussian framework.

\fangzhen{For a single-mode Gaussian channel, the condition as in Eq.~\eqref{positive_condition} can be simplified as follows.
One can check that for an arbitrary real \(2\times2\) matrix
$X\Omega X^{T}=(\det X)\Omega$.
Hence, Eq.~\eqref{positive_condition} becomes
$\det\!\left[
Y+i(1-\det X)\Omega
\right]
\geq 0.$
Therefore,
$\det Y\geq(1-\det X)^2$,
which is equivalently written as
\begin{equation} \label{detY_detX_condition}
\sqrt{\det Y}\geq\left|1-\det X\right|.
\end{equation}
In other words, any single-mode Gaussian channel $(X,Y)$ should satisfy Eq.~\eqref{detY_detX_condition}. This corresponds to the shaded region in Fig.~\ref{fig:Gaussian}.
}

\fangzhen{

Single-mode Gaussian covariant channels with $\det X<1$ are loss channels. These channels can be realized by mixing the system mode with an auxiliary mode initially in Gibbs state via a beam-splitter. As shown in Appendix~\ref{app_B}, every single-mode Gaussian loss channel, except those on the quantum-limited loss boundary, admits a Gaussian dilation that is both phase covariant and tabletop time reversible. These channels correspond to the green area in Fig.~\ref{fig:Gaussian}.

By contrast, as shown in Theorem~\ref{theorem_2}, amplification channels, which correspond to those with $\det X>1$, do not admit a Gaussian TTR dilation. These channels correspond to the red area in Fig.~\ref{fig:Gaussian}.


For channels with $\det X=1$, which correspond to the
red line in Fig.~\ref{fig:Gaussian}, we show that no nontrivial Gaussian
TTR dilation exists. 
Setting $\tau=1$ in Eqs.~\eqref{eq: backward U}--\eqref{X_rev_amplification}, we find that the tabletop reverse channel satisfies
$\det X_{\mathrm{rev}}=1.$ 
However, for any nonzero phase-covariant Gaussian noise $Y$,
$
V_{\alpha'}=XV_\alpha X^T+Y
$
implies $\det V_{\alpha'}>\det V_\alpha$. According to Eq.~\eqref{XP},
\begin{equation}
\det X_{\mathrm P}
=
\frac{\det V_\alpha-1}{\det V_{\alpha'}-1}
<1
=
\det X_{\mathrm{rev}}.
\end{equation}
Therefore, a Gaussian TTR dilation exists only if $Y=0$. It is easy to check that channels with  
$
\bigl(\sqrt{\det Y},\det X\bigr)=(0,1)
$
are unitary channels.
}

\begin{figure}[t]
    \centering
    \includegraphics[width=\columnwidth]{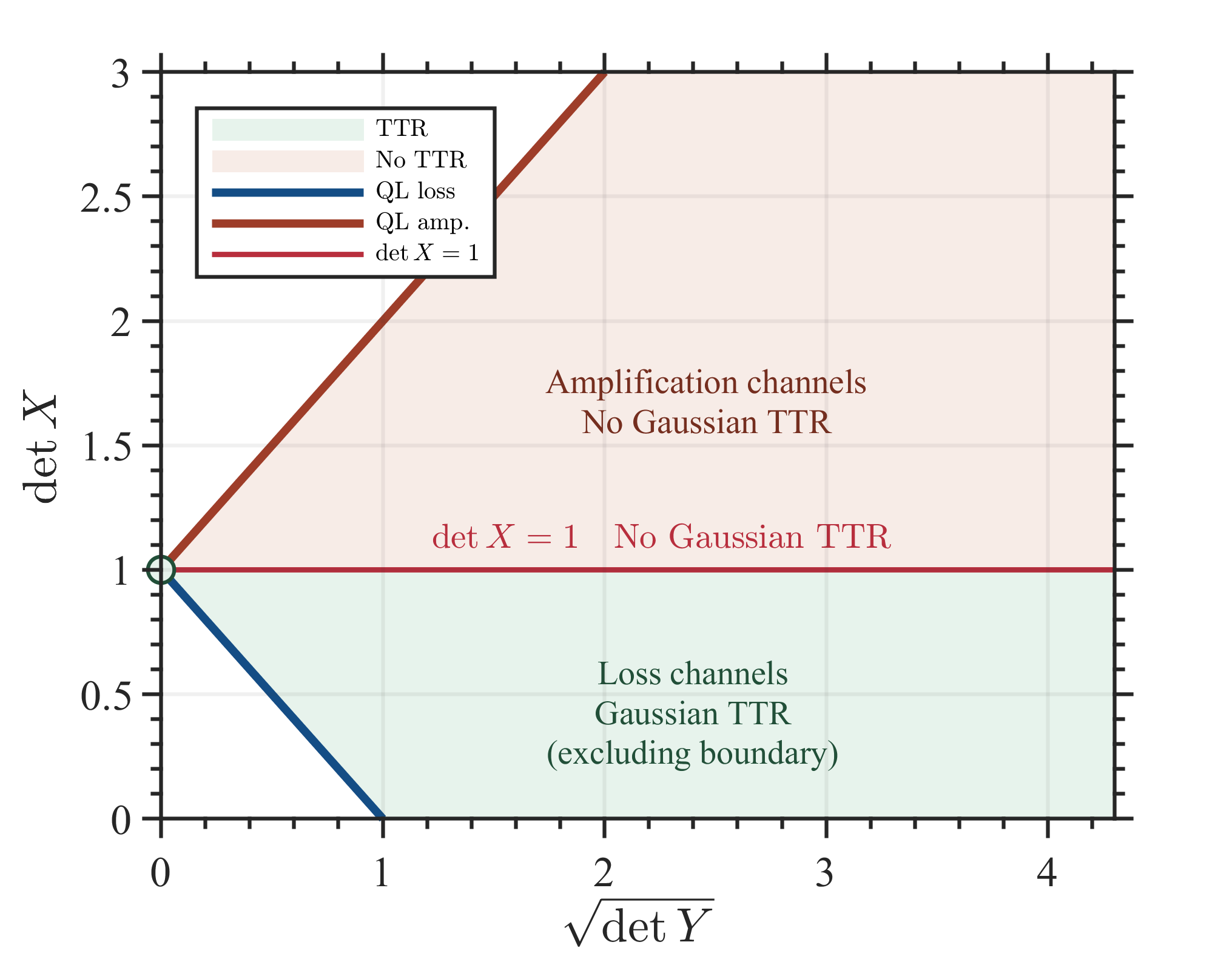}
    \caption{
    \fangzhen{Schematic illustration of TTR for single-mode phase-covariant Gaussian dilations. Complete positivity requires $\sqrt{\det Y}\geq |1-\det X|$, so that the physically allowed channels lie between the quantum-limited loss boundary $\det X=1-\sqrt{\det Y}$ and the quantum-limited amplification boundary $\det X=1+\sqrt{\det Y}$. The shaded region with $\det X<1$ represents loss channels that admit Gaussian dilations realizing TTR, excluding the quantum-limited loss boundary. By contrast, amplification channels with $\det X>1$ do not admit Gaussian dilations realizing TTR. The horizontal line $\det X=1$ admits TTR only at the point $(0,1)$, corresponding to the
trivial Gaussian unitary case.} 
    }
    \label{fig:Gaussian}
\end{figure}

\section{Conclusions and Outlook}
This paper analyzes TTR of phase-covariant operations. Finite-dimensional and Gaussian systems lead to markedly different conclusions.
In finite-dimensional systems, for phase-covariant operations that admit a covariant prior state $\alpha$ with full-rank output, we show that one can choose this state as the prior and construct a covariant dilation that is tabletop time-reversible. 
More generally, we show that every finite-dimensional CPTP map with a full-rank output state admits a dilation realizing TTR. 
In contrast, for phase-covariant Gaussian \fangzhen{systems}, the single-mode Gaussian amplification channel provides an explicit example showing that there are Gaussian phase-covariant operations that do not admit a Gaussian dilation that achieves TTR. 
Our results provide a method for constructing devices that achieve TTR for both covariant channels and general CPTP maps, and further reveal the Gaussian constraint on achieving TTR.

Finally, we identify two open questions and possible directions for future work.

First, the symmetry studied in this work corresponds only to the
time-translation symmetry generated by the system Hamiltonian, namely the
$U(1)$-covariant case. Our finite-dimensional dilation relies on this
specific symmetry through the Bohr-frequency decomposition and the associated
energy-preserving covariant dilation. It is therefore natural to ask whether, for channels covariant with respect to more general \fangzhen{symmetry groups $G$~\cite{marvian2012symmetry}}, one can still construct a covariant
dilation that realizes TTR.

Second, our counterexample in the continuous-variable setting is established only within
the Gaussian framework. This leaves open the possibility that TTR
might still be realizable beyond the Gaussian framework. Moreover, continuous-variable systems also admit Kraus-operator representations, although the number of Kraus operators is generally infinite~\cite{ivan2011operator}. This leaves open the possibility that the construction of Theorem~\ref{theorem} may be extended to the continuous-variable setting, although the infinite-dimensional nature of the problem and the potential occurrence of divergences may create significant technical challenges.
Hence, it would be interesting to study whether phase-covariant operations in continuous-variable systems can achieve TTR by allowing non-Gaussian dilations.

\begin{acknowledgments}
This work is supported by Shandong Provincial Quantum Joint Fund, under Grants No.ZR2025LLZ002 and No.ZR2025LLZ001.
\end{acknowledgments}

\appendix

\section{Proof of Lemma 1} \label{app_A}

Our proof is based on two facts. 
Firstly, $M$ is unitary if and only if $M'=U_1MU_2^\dagger$ is unitary, where $U_1$ and $U_2$ are two unitary matrices. 
Secondly, if the specified entries of $M'$ form a unitary operator on a subspace, the complement space of which is a subset of the space on which the unspecified entries act, then the unspecified entries can be chosen such that $M'$ is a unitary matrix.
For example, consider a square matrix $M'$ that can be written as
\begin{equation}\label{eq:M prime}
    M'=U_s+M_{u}',
\end{equation} 
where $U_s$ is a unitary acting on a subspace spanned by $\{|i\rangle\}_{i=1}^{2m}$ and $M_{u}'$ is an unspecified matrix acting on a subspace that is a superset of $\mathrm{span}\{|i\rangle\}_{i=2m+1}^N$. 
Then one may choose $M_u'=\mathbf{0}\oplus U_u$, where $\mathbf{0}$ is a zero matrix acting on $\mathrm{span}\{|i\rangle\}_{i=m+1}^{2m}$ and $U_u$ is an arbitrary unitary matrix acting on $\mathrm{span}\{|i\rangle\}_{i=2m+1}^{N}$. 
It follows that $M'$ becomes $M'=U_s\oplus U_u$, which is apparently a unitary matrix.

\fangzhen{
Next, we construct unitary matrices $U_1$ and $U_2$ such that $M'$ has the form of Eq.~\eqref{eq:M prime}.
To start with, based on Eqs.~\eqref{eq:orth}~and~\eqref{zero}, we check that both $\bigl\{|i\rangle\bigr\}_{i\in\mathcal J_g}\cup\bigl\{|f_j\rangle\bigr\}_{j\in\mathcal J_f}$ and $\bigl\{|j\rangle\bigr\}_{j\in\mathcal J_f}\cup\bigl\{|g_i\rangle\bigr\}_{i\in\mathcal J_g}$ are sets of $2m$ orthonormal vectors. Therefore, one can construct orthonormal vectors $\{|u_k\rangle\}_{k=2m+1}^{N}$ and $\{|v_k\rangle\}_{k=2m+1}^{N}$, such that both $
\bigl\{|i\rangle\bigr\}_{i\in\mathcal J_g}
\cup
\bigl\{|f_j\rangle\bigr\}_{j\in\mathcal J_f}
\cup
\bigl\{|u_k\rangle\bigr\}_{k=2m+1}^{N}
$
and
$
\bigl\{|j\rangle\bigr\}_{j\in\mathcal J_f}
\cup
\bigl\{|g_i\rangle\bigr\}_{i\in\mathcal J_g}
\cup
\bigl\{|v_k\rangle\bigr\}_{k=2m+1}^{N}
$
form complete orthonormal bases of \(\mathbb C^N\).

Next, we define $U_1$ and $U_2$ as
\begin{align}
U_1
={}&
\sum_{i\in\mathcal J_g}
|\bar i\rangle\langle i|
+
\sum_{j\in\mathcal J_f}
|m+\bar j\rangle\langle f_j|
+
\sum_{k=2m+1}^{N}
|k\rangle\langle u_k|
\label{eq:explicit_U1}
\end{align}
and
\begin{align}
U_2
={}&
\sum_{j\in\mathcal J_f}
|\bar j\rangle\langle j|
+
\sum_{i\in\mathcal J_g}
|m+\bar i\rangle\langle g_i|
+
\sum_{k=2m+1}^{N}
|k\rangle\langle v_k|,
\label{eq:explicit_U2}
\end{align}
where $\bar i=1,\dots,m$ is the order number of $i$ in the set $\mathcal{J}_g$ and $\bar j=1,\dots,m$ is the order number of $j$ in the set $\mathcal{J}_f$.
It follows that
\begin{align}
        M' & =  U_1 M U_2^\dagger \nonumber\\
    & = \sum_{ j=1}^m |m+ j\rangle\langle j|+\sum_{ i=1}^m | i\rangle\langle m+ i| + M_u',
\end{align}
where
\begin{align}
        M_u'&=\sum_{\substack{i\in\mathcal J_g^c, j\in\mathcal J_f^c}}
M_{ij}U_1|i\rangle\langle j|U_2^\dagger\nonumber\\
& =  \sum_{i,j=m+1}^{N} M'_{ij}\,|i\rangle\langle j|.
\end{align}
The reason for the second equation is as follows. Since $\mathrm{span}\{|i\rangle\}_{i\in\mathcal{J}_g}$ is orthogonal to $\mathrm{span}\{|i\rangle\}_{i\in\mathcal{J}_g^c}$, the subspace $\mathrm{span}\{U_1|i\rangle\}_{i\in\mathcal{J}_g}$ is orthogonal to $\mathrm{span}\{U_1|i\rangle\}_{i\in\mathcal{J}_g^c}$. 
Therefore, an orthonormal basis of $\mathrm{span}\{U_1|i\rangle\}_{i\in\mathcal{J}_g^c}$ can be chosen as $\{|i\rangle\}_{i=m+1}^N$. Similarly, the basis of $\mathrm{span}\{U_2|j\rangle\}_{j\in\mathcal{J}_f^c}$ can be chosen as $\{|j\rangle\}_{j=m+1}^N$.

Now we say $M'$ is in the form of Eq.~\eqref{eq:M prime} for the following reasons. 
Firstly, the specified part $U_s=\sum_{ j=1}^m |m+ j\rangle\langle j|+\sum_{ i=1}^m | i\rangle\langle m+ i|$ is a unitary acting on a subspace spanned by $\{|i\rangle\}_{i=1}^{2m}$.
Secondly, the unspecified part $M_u'$ acts on a subspace spanned by $\{|i\rangle\}_{i=m+1}^N$, which is a superset of $\mathrm{span}\{|i\rangle\}_{i=2m+1}^{N}$.
Therefore, one may choose $M_u'$ such that $M'=U_s\oplus U_u$. It follows that
\begin{align}\label{eq:explicit_M_without_B}
M
&=
U_1^\dagger
\left(
U_s\oplus U_u
\right)
U_2\nonumber\\
&={}
\sum_{j\in\mathcal J_f}
|f_j\rangle\langle j|
+
\sum_{i\in\mathcal J_g}
|i\rangle\langle g_i|
+
\sum_{k,l=2m+1}^{N}
\langle k|U_u|l\rangle
|u_k\rangle\langle v_l|.
\end{align}

Since every \(|u_k\rangle\) is orthogonal to
\(\{|i\rangle\}_{i\in\mathcal J_g}\), it has support only on
\(\operatorname{span}\{|i\rangle\}_{i\in\mathcal J_g^c}\).
Similarly, every \(|v_l\rangle\) has support only on
\(\operatorname{span}\{|j\rangle\}_{j\in\mathcal J_f^c}\).
Thus, comparing Eq.~\eqref{eq:explicit_M_without_B} with
Eq.~\eqref{eq:partial-M}, the unspecified entries are
\begin{equation}\label{eq:general_explicit_Mij}
M_{ij}
=
\sum_{k,l=2m+1}^{N}
\langle k|U_u|l\rangle
\langle i|u_k\rangle
\langle v_l|j\rangle,
\qquad
i\in\mathcal J_g^c,\quad
j\in\mathcal J_f^c.
\end{equation}

For simplicity, we choose $U_u=\mathbb{I}_{N-2m}$ and obtain
\begin{equation}\label{eq:explicit_Mij_without_U}
M_{ij}
=
\sum_{k=2m+1}^{N}
\langle i|u_k\rangle
\langle v_k|j\rangle,
\quad
i\in\mathcal J_g^c,\
j\in\mathcal J_f^c.
\end{equation}
Equivalently,
\begin{equation}\label{eq:explicit_M_final}
M
=
\sum_{j\in\mathcal J_f}
|f_j\rangle\langle j|
+
\sum_{i\in\mathcal J_g}
|i\rangle\langle g_i|
+
\sum_{k=2m+1}^{N}
|u_k\rangle\langle v_k|.
\end{equation}
In other words, when we choose the unspecified elements in $M$ as in Eq.~\eqref{eq:explicit_Mij_without_U}, the resulting matrix $M$ is unitary.
}

\section{Tabletop Reversibility of a Single-Mode Gaussian Loss Channel} \label{app_B}
A single-mode Gaussian loss channel can be \fangzhen{described} as follows~\cite{weedbrook2012gaussian}
\begin{equation} \label{appendix_B_loss_channel}
V \longmapsto \tau V +(1 - \tau)(2\overline{n} + 1)I,\ {\vec d} \longmapsto \sqrt{\tau} {\vec d},\  0<\tau  <1\end{equation}
where $\overline{n}$ is a thermal number, which is nonnegative. The thermal state we will use is $
V_{\mathrm{th}} = (2\bar{n}+1)\, I
$ and ${\vec d_{th}}=0$. Using Eq.~\eqref{eq:gaussian_petz_recovery}, and assuming that both the bath state and the prior are thermal states, we derive the Petz recovery map in Gaussian form as
\begin{widetext}
    \begin{equation}\label{loss_eq1}
V \longmapsto \frac{4  \overline{n}_{2}  (1 + \overline{n}_{2}) \tau}{-1 + \bigl( \tau  (1 + 2  \overline{n}_{2}) + (1 - \tau)  (1 + 2  \overline{n}_{1}) \bigr)^2}V + ((2  \overline{n}_{2} + 1) - \frac{  \tau  \bigl( (2 \overline{n}_{2} + 1)^2 - 1 \bigr) \bigl( \tau  (2 \overline{n}_{2} + 1) + (1 - \tau)  (2 \overline{n}_{1} + 1) \bigr) }{ \bigl( \tau  (2  \overline{n}_{2} + 1) + (1 - \tau)  (2  \overline{n}_{1} + 1) \bigr)^2 - 1 }) I \\
\end{equation}
where $\overline{n}_1$ is the thermal number of the bath state of the forward channel and $\overline{n}_2$ is the thermal number of the prior state.
The Gaussian unitary of the Gaussian loss channel is~\cite{weedbrook2012gaussian}
\begin{equation}
    S=
\begin{pmatrix}
\sqrt{\tau I} & \sqrt{1-\tau }I\\
-\sqrt{1-\tau }I & \sqrt{\tau I}
\end{pmatrix}.
\end{equation}
As in Eq.~\eqref{eq:reverse}, we find that if a single-mode Gaussian loss channel has a tabletop reverse form, it is given by 
\begin{equation}\label{loss_eq2}
V \longmapsto \tau V + (1 - \tau)(2\overline{n}_{3} + 1)I
\end{equation}
where we choose its reverse bath state to also be a thermal state with thermal number  $\overline{n}_3$. \\
\fangzhen{If a single-mode Gaussian loss channel is tabletop time reversible, Eq.~\eqref{loss_eq1} and Eq.~\eqref{loss_eq2}  must be equal.} So it means that the coefficients of $V$ must be equivalent, and the same applies to $I$ as well. Therefore, we need 
\begin{equation}\label{C4}
\begin{aligned}
\tau
&=
\frac{4 \overline{n}_{2}(1+\overline{n}_{2})\tau}
{-1+\bigl(\tau(1+2\overline{n}_{2})+(1-\tau)(1+2\overline{n}_{1})\bigr)^2},
\\[0.5ex]
(1-\tau)(2\overline{n}_{3}+1)
&=
(2\overline{n}_{2}+1)
-\frac{
\tau\bigl((2\overline{n}_{2}+1)^2-1\bigr)
\bigl(\tau(2\overline{n}_{2}+1)+(1-\tau)(2\overline{n}_{1}+1)\bigr)
}{
\bigl(\tau(2\overline{n}_{2}+1)+(1-\tau)(2\overline{n}_{1}+1)\bigr)^2-1
}.
\end{aligned}
\end{equation}
\end{widetext}

Then we find that the only physically meaningful solution to the equation is $\overline{n}_1=\overline{n}_2=\overline{n}_3\neq 0$. By imposing the TTR constraint on the displacement vector ${\vec d}$, we obtain the same condition as Eq.~\eqref{C4}. So a single-mode Gaussian loss channel is tabletop time reversible when its prior state and environment state are the same thermal state with $\bar n\neq 0$.

For a faithful thermal prior, $\overline{n}_2>0$. From the first
equality in Eq.~\eqref{C4}, we obtain
\begin{equation}
\left[
\tau(1+2\overline{n}_2)
+(1-\tau)(1+2\overline{n}_1)
\right]^2
=
(1+2\overline{n}_2)^2.
\end{equation}
Since both sides have positive square roots and $0<\tau<1$, this
implies
\begin{equation}
\overline{n}_1=\overline{n}_2.
\end{equation}
Substituting this relation into the second equality in
Eq.~\eqref{C4} gives
\begin{equation}
(1-\tau)(2\overline{n}_3+1)
=
(1-\tau)(2\overline{n}_2+1),
\end{equation}
and hence $\overline{n}_3=\overline{n}_2$. Therefore,
\begin{equation}
\overline{n}_1=\overline{n}_2=\overline{n}_3.
\end{equation}
Thus, TTR requires the prior state to be the same thermal state as
the forward environment state.

\fangzhen{Finally, we explain the exclusion of the quantum-limited loss boundary in Fig.~\ref{fig:Gaussian}, which corresponds to $\bar n=0$ in Eq.~\eqref{appendix_B_loss_channel}. For the dilation considered above, the preceding relations imply $\bar n_1=\bar n_2=\bar n_3=0$, causing the denominators in Eq.~\eqref{loss_eq1} to vanish. This singularity reflects a general obstruction: TTR for a single-mode loss channel requires the prior to coincide with the corresponding thermal state. On the quantum-limited loss boundary, the prior needs to be a zero-temperature thermal state, which is non-faithful, and hence the Petz recovery map is not well defined. A detailed proof is given below.

Using Eq.~\eqref{XP} and
$(V\Omega)^2=-\det(V)I$, the linear part of the Petz recovery map is
\begin{equation} \label{Xp_loss_channel_appendix_B}
X_{\mathrm P}
=
\sqrt{\tau}\,\sqrt{
\frac{1-(\det V_\alpha)^{-1}}
{1-(\det V_{\alpha'})^{-1}}
}\,
V_\alpha V_{\alpha'}^{-1}.
\end{equation}
Moreover, for any
Gaussian dilation of the loss channel, the expression for $X_{\mathrm{rev}}$ in
Eq.~\eqref{X_rev_amplification} remains valid.
The TTR condition $X_{\mathrm P}=X_{\mathrm{rev}}$ therefore implies
$
V_{\alpha'}=V_\alpha.
$
From Eq.~\eqref{appendix_B_loss_channel}, we obtain
\begin{equation}\label{eq:loss_prior_bath}
V'_\alpha=V_\alpha=(2\overline{n} + 1)I.
\end{equation}
On the quantum-limited loss boundary, $\overline{n}=0$, so that
$V'_\alpha=V_\alpha=I$, corresponding to the nonfaithful vacuum state. Consequently, the Petz recovery map is not well defined.}

\section{A Worked Example: Explicit Construction of a TTR Dilation for a Qutrit EnTO}\label{Worked Example}

This appendix presents the detailed construction of the unitary used in the qutrit example in the main text, together with the corresponding numerical procedure. All assumptions and conditions are the same as those specified in the qutrit example.

\begin{figure*}
\includegraphics[width=16cm]{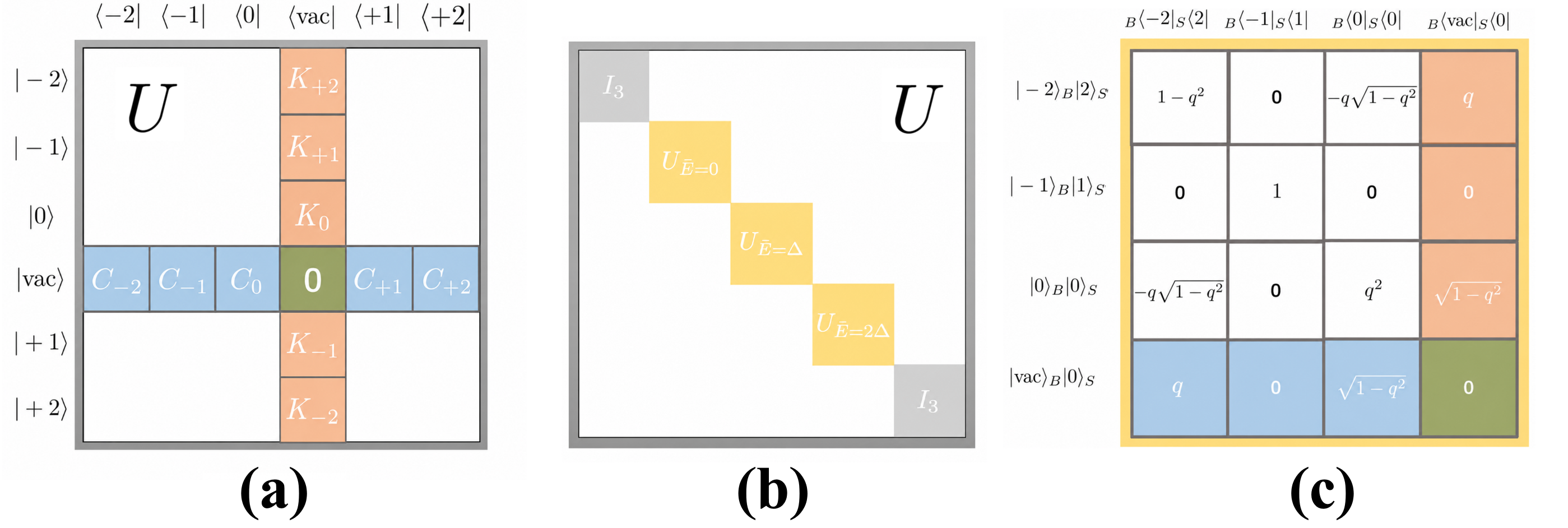}
\caption{\label{TABLETOP/Qutrit_example.png}
\fangzhen{Schematic illustration of the unitary construction in the qutrit example.
(a) Fixed column and row according to the Kraus operators and Petz recovery map.
(b) Block decomposition of the global unitary according to total energy, \(U=\bigoplus_{\bar E} U_{\bar E}\). 
(c) Structure of a constrained energy subblock \(U_{\bar E=0}\).}}
\end{figure*}

\textbf{(a) Fixed column and row according to Kraus operators and Petz recovery map.}
According to Eq.~\eqref{reverse_map_kraus}, the Kraus operators of the Petz recovery map are
\begin{equation}
    C_\omega^\dagger
    =
    (\gamma_S^\kappa)^{1/2}
    K_\omega^\dagger
    (\mathcal E(\gamma_S^\kappa))^{-1/2}=(\gamma_S^\kappa)^{1/2}
    K_\omega^\dagger
    (\gamma_S^\kappa)^{-1/2}.
\end{equation}
Since \(\gamma_S^\kappa\) is diagonal in the energy basis, one obtains
\begin{equation}
    C_{+2}^\dagger=|0\rangle\langle 2|=K_{-2},
    \
    C_0^\dagger=K_0,
    \
    C_{-2}^\dagger=q|2\rangle\langle 0|=K_{+2}.
\end{equation}
Hence,
\begin{equation}
    C_\omega^\dagger=K_{-\omega}.
\end{equation}
Therefore, the Petz map is explicitly
\begin{equation}
    \widehat{\mathcal E}_{\gamma_S^\kappa}(\rho)
    =
    K_{-2}\rho K_{-2}^\dagger
    +
    K_0\rho K_0^\dagger
    +
    K_{+2}\rho K_{+2}^\dagger .
\end{equation}

We now construct the tabletop-reversible covariant dilation. 
Following Eqs.~\eqref{unitary construction column} and~\eqref{unitary construction block}, the forward channel fixes the
\(|{\rm \mathrm{vac}}\rangle\) column of the global unitary as
\begin{equation}
\begin{aligned}
    \langle -2|_B U |{\rm \mathrm{vac}}\rangle_B &= K_{+2},\\
    \langle 0|_B U |{\rm \mathrm{vac}}\rangle_B &= K_{0},\\
    \langle +2|_B U |{\rm \mathrm{vac}}\rangle_B &= K_{-2},\\
    \langle -1|_B U |{\rm \mathrm{vac}}\rangle_B &= K_{+1}=0,\\
    \langle +1|_B U |{\rm \mathrm{vac}}\rangle_B &= K_{-1}=0,\\
    \langle {\rm \mathrm{vac}}|_B U |{\rm \mathrm{vac}}\rangle_B &=0 .
\end{aligned}
\end{equation}

From~\eqref{unitary construction row}, the Petz recovery map fixes the \(\langle{\rm \mathrm{vac}}|\) row of the global unitary as
\begin{equation}
\begin{aligned}
    \langle {\rm \mathrm{vac}}|_B U |+2\rangle_B &= C_{+2}=K_{-2}^\dagger,\\
    \langle {\rm \mathrm{vac}}|_B U |0\rangle_B &= C_{0}=K_{0},\\
    \langle {\rm \mathrm{vac}}|_B U |-2\rangle_B &= C_{-2}=K_{+2}^\dagger,\\
    \langle {\rm \mathrm{vac}}|_B U |+1\rangle_B &= C_{+1}=0,\\
    \langle {\rm \mathrm{vac}}|_B U |-1\rangle_B &= C_{-1}=0,\\
    \langle {\rm \mathrm{vac}}|_B U |{\rm \mathrm{vac}}\rangle_B &=0 .
\end{aligned}
\end{equation}
Thus, the structure of the global unitary is shown in Fig.~~\ref{TABLETOP/Qutrit_example.png}(a). 

\textbf{(b) Block decomposition of the global unitary according to total energy.} After a change of representation, which is merely a reordering of the basis vectors, the prescribed unitary elements are arranged into a block-diagonal form according to the total energy \(\bar E\) as in~\eqref{eq: U block}. 
Then, we get Fig.~\ref{TABLETOP/Qutrit_example.png}(b). 
For total-energy sectors other than \(\bar E=0,\Delta,2\Delta\), the corresponding blocks of \(U\) are unconstrained and may be completed by arbitrary unitary operators. For convenience, we choose them as identity $I$.
The only nontrivial constrained blocks are those with
\(\bar E=0,\Delta,2\Delta\). The corresponding eigenspaces are spanned by
\(\mathcal H_{\bar E=0}
=
{\rm span}\{
|-2\rangle_B|2\rangle_S,
|-1\rangle_B|1\rangle_S,
|0\rangle_B|0\rangle_S,
|{\rm \mathrm{vac}}\rangle_B|0\rangle_S
\}\),
\(\mathcal H_{\bar E=\Delta}
=
{\rm span}\{
|-1\rangle_B|2\rangle_S,
|0\rangle_B|1\rangle_S,
|{\rm \mathrm{vac}}\rangle_B|1\rangle_S,
|+1\rangle_B|0\rangle_S
\}\),
and
\(\mathcal H_{\bar E=2\Delta}
=
{\rm span}\{
|0\rangle_B|2\rangle_S,
|{\rm vac}\rangle_B|2\rangle_S,
|+1\rangle_B|1\rangle_S,
|+2\rangle_B|0\rangle_S
\}\).

Therefore, the partially specified submatrices, with one row and one column already fixed and the remaining entries to be completed, are given by
\begin{align}
U_{\bar E=0}
&=
\begin{pmatrix}
? & ? & ? & q\\
? & ? & ? & 0\\
? & ? & ? & \sqrt{1-q^2}\\
q & 0 & \sqrt{1-q^2} & 0
\end{pmatrix},
\\
U_{\bar E=\Delta}
&=
\begin{pmatrix}
? & ? & 0 & ?\\
? & ? & 1 & ?\\
0 & 1 & 0 & 0\\
? & ? & 0 & ?
\end{pmatrix},
U_{\bar E=2\Delta}
=
\begin{pmatrix}
? & 0 & ? & ?\\
0 & 0 & 0 & 1\\
? & 0 & ? & ?\\
? & 1 & ? & ?
\end{pmatrix}, 
\label{eq:U_blocks_unspecified}
\end{align}
where the symbols ``$?$'' denote unspecified entries.

\textbf{(c) Structure of a constrained energy subblock \(U_{\bar E}\).} We now explicitly complete the three constrained blocks by applying
Eq.~\eqref{eq:explicit_Mij}.
For each of the three blocks, we denote by
$\{|e_j\rangle\}_{j=1}^4$ the standard basis associated with the
corresponding ordering of the total-energy eigenspace. For $\bar E=0$, the prescribed fourth column and fourth row are
characterized by
\begin{equation}
\
|f_0\rangle
=
|g_0\rangle
=
q|e_1\rangle
+
\sqrt{1-q^2}|e_3\rangle.
\end{equation}

We can choose
\begin{equation}
|u_3\rangle
=
|v_3\rangle
=
\sqrt{1-q^2}|e_1\rangle-q|e_3\rangle,
\end{equation}

\begin{equation}
|u_4\rangle
=
|v_4\rangle
=
\ |e_2\rangle,
\end{equation}
which form an orthonormal basis of
$\operatorname{span}
\left\{
|e_4\rangle,\,
q|e_1\rangle+\sqrt{1-q^2}|e_3\rangle
\right\}^{\perp}$.

We choose the unitary acting on this two-dimensional complementary
subspace as
\begin{equation}
U_u^{(0)}=I_2.
\end{equation}
Substituting these choices into
Eqs.~\eqref{eq:partial-M} and \eqref{eq:explicit_Mij}, we obtain
\begin{align}
U_{\bar E=0}
={}&
\left(
q|e_1\rangle+\sqrt{1-q^2}|e_3\rangle
\right)\langle e_4|
\nonumber\\
&+
|e_4\rangle
\left(
q\langle e_1|
+\sqrt{1-q^2}\langle e_3|
\right)
\nonumber\\
&+
\left(
\sqrt{1-q^2}|e_1\rangle-q|e_3\rangle
\right)
\left(
\sqrt{1-q^2}\langle e_1|-q\langle e_3|
\right)
\nonumber\\
&+
|e_2\rangle\langle e_2|
\nonumber\\
={}&
\begin{pmatrix}
1-q^2 & 0 & -q\sqrt{1-q^2} & q\\
0 & 1 & 0 & 0\\
-q\sqrt{1-q^2} & 0 & q^2 & \sqrt{1-q^2}\\
q & 0 & \sqrt{1-q^2} & 0
\end{pmatrix}.
\label{eq:U_E0_completion}
\end{align}
With the same method, we can also complete $U_{\bar E=\Delta}$ and $U_{\bar E=2\Delta}$  as in Eq.~\eqref{qutrit_example_result}.

\section{Non-Gaussian Tabletop Realization of a Single-Mode Gaussian Amplifier}
\label{app:asymptotic_nonGaussian_TTR_amp}

In this appendix, we construct a non-Gaussian dilation of a single-mode Gaussian amplification channel, which is both covariant and tabletop time reversible. It indicates that Gaussianity imposes highly restrictive structural constraints on dilations and recovery processes. 

In Theorem~\ref{theorem_2}, we have shown that the single-mode Gaussian amplification channel \(\mathcal{E}_{A}\) admits no Gaussian dilation whose tabletop reverse channel coincides with the corresponding Petz recovery map. This result naturally raises the question of whether the obstruction comes from the continuous-variable nature of the system itself or instead from the restriction to Gaussian dilations. In other words, it remains possible that a non-Gaussian dilation of the same amplification channel may realize TTR.

The motivation for considering this possibility can be traced back to the key contradiction in the proof of Theorem~\ref{theorem_2}. For the amplification channel \(\mathcal{E}_{A}\), Eq.~\eqref{detXP} gives $\det X_\mathrm{P}<1$, which means that the Petz recovery map acts as a loss channel. By contrast, as shown in the proof of Theorem~\ref{theorem_2}, the system block of the tabletop reverse channel is still $X_\mathrm{rev}
=
\sqrt{\tau}\,I,
\
\det X_\mathrm{rev}=\tau>1$. Thus, the tabletop reverse channel generated by a Gaussian dilation is
still an amplification channel, whereas the Petz recovery map is
a loss channel. This mismatch is the essential obstruction to Gaussian TTR for the amplification channel.

Physically, this obstruction can be understood as follows. A Gaussian amplification dilation is implemented by a Gaussian unitary, such as a two-mode squeezing interaction, whose operation is supported by an external energy supply. Reversing such a Gaussian unitary changes the direction of the squeezing transformation, but the induced tabletop reverse channel on the system still has an amplifying linear part. Therefore, the reverse process available within the same Gaussian device architecture is not an energy-reducing channel, and hence cannot realize the Petz recovery map.

This observation suggests a natural way to bypass the Gaussian obstruction. If one allows a non-Gaussian dilation, it may be possible to design the same tabletop reverse architecture so that the reversed unitary induces an effective energy-decreasing channel on the system. In that case, the tabletop reverse channel can become a loss channel and may therefore coincide with the Petz recovery map of the amplification channel.

In the following, we show that this idea can indeed be realized. We construct a three-mode non-Gaussian dilation whose forward channel converges to the single-mode Gaussian amplification channel, while its tabletop reverse channel converges to the Petz recovery map of that amplification channel. Hence, although exact Gaussian TTR is forbidden by Theorem~\ref{theorem_2}, the single-mode Gaussian amplification channel admits a non-Gaussian tabletop time-reversible dilation.

The construction involves a single-mode system and a two-mode
environment. The same three-mode unitary is used in the forward and backward processes, while the environmental state is allowed to be changed in the tabletop reverse. The essential point is that in different limits the same trilinear interaction reduces either to a two-mode squeezing interaction or to a beam-splitter interaction.

Let \(a\) denote the annihilation operator of the system mode and let \(b,c\) denote the annihilation operators of two environmental modes. We consider the three-mode Hamiltonian
\begin{equation}\label{eq:trilinear_H_plus}
H_3
=
\hbar g_3
\left(
c^\dagger a b
+
c a^\dagger b^\dagger
\right),
\end{equation}
where \(g_3\in\mathbb R\) is the three-mode coupling strength. The corresponding unitary is
\begin{equation}\label{eq:trilinear_U_plus}
U_\varepsilon
=
\exp\left[
-i\varepsilon
\left(
c^\dagger a b
+
c a^\dagger b^\dagger
\right)
\right],
\qquad
\varepsilon:=g_3 t .
\end{equation}
Since the generator is cubic in the field operators, \(U_\varepsilon\) is a non-Gaussian unitary.

The target forward channel is the phase-insensitive single-mode Gaussian amplification channel as Eq.~\eqref{Gaussian_amplification}.

\textbf{Forward channel.}
This subsection shows that the Gaussian amplification channel can be realized by the non-Gaussian dilation defined in Eq.~\eqref{eq:trilinear_U_plus}.
To ensure that the channel in this example remains phase covariant, we choose the environmental state of mode $c$ as follows:
\begin{equation} \label{C_state}
\omega_{\mathrm{F},c}^{(\varepsilon)}
:=
\int_0^{2\pi}
\frac{d\phi}{2\pi}
\left|
\frac{r_\mathrm{F}}{\varepsilon}e^{i\phi}
\right\rangle_c
\left\langle
\frac{r_\mathrm{F}}{\varepsilon}e^{i\phi}
\right|,
\
r_\mathrm{F}:=\operatorname{arcosh}\sqrt{\tau}.
\end{equation}
Using the Fock expansion of coherent states, this state can be written as
\begin{equation}
\omega_{\mathrm{F},c}^{(\varepsilon)}
=
\exp\left[-\frac{r_\mathrm{F}^2}{\varepsilon^2}\right]
\sum_{n=0}^{\infty}
\frac{
\left(r_\mathrm{F}^2/\varepsilon^2\right)^n
}{n!}
|n\rangle_c\langle n|.
\end{equation}
For each mode label \(\ell\in\{a,b,c\}\), we denote the photon-number operator by
\(N_\ell:=\ell^\dagger \ell\) and the free Hamiltonian by
\(H_\ell:=\hbar\omega_\ell N_\ell\).
The photon-number operator of mode \(c\) is
\(N_c:=c^\dagger c\). In the Fock basis, it takes the form
\(
N_c=\sum_{n=0}^{\infty}n\,|n\rangle_c\langle n|\). Therefore,
$[\omega_{\mathrm{F},c}^{(\varepsilon)},N_c]=0$. We choose the forward environmental state as
\begin{equation}\label{duiyi1}
\eta_\mathrm{F}^{(\varepsilon)}
=
\rho_\mathrm{th}^{(b)}(\bar n)
\otimes
\omega_{\mathrm{F},c}^{(\varepsilon)}.
\end{equation}
Since the thermal state is also diagonal in the number basis, one has
\begin{equation}\label{duiyi2}
[\eta_\mathrm{F}^{(\varepsilon)},H_b+H_c]=0.
\end{equation}
Moreover, under the resonance condition \(\omega_c=\omega_a+\omega_b\),
the Hamiltonian
$H_3
=
\hbar g_3
\left(
c^\dagger ab+ca^\dagger b^\dagger
\right)$
satisfies $[H_3,H_a+H_b+H_c]=0$. Therefore, 
\begin{equation} \label{duiyi3}
    [U_\varepsilon,H_a+H_b+H_c]=0.
\end{equation}

According to Proposition~2 of Ref.~\cite{marvian2016quantify}, the conditions
\([\eta_F^{(\varepsilon)},H_b+H_c]=0\) and
\([U_\varepsilon,H_a+H_b+H_c]=0\) ensure that
the forward channel
\begin{equation} \label{Forward_channel_qutrit}
\mathcal{E}_\mathrm{F}^{(\varepsilon)}(\rho_a)
=
\operatorname{Tr}_{bc}
\left[
U_\varepsilon
\left(
\rho_a\otimes\eta_\mathrm{F}^{(\varepsilon)}
\right)
U_\varepsilon^\dagger
\right]
\end{equation}
is phase covariant.

Because the environmental state is a phase average of coherent states, the
forward channel can be written as
\begin{equation}
\mathcal{E}_\mathrm{F}^{(\varepsilon)}
=
\int_0^{2\pi}
\frac{d\phi}{2\pi}
\mathcal{E}_{\mathrm{F},\phi}^{(\varepsilon)} .
\end{equation}
For each fixed \(\phi\), let
\begin{equation}
\zeta_{\mathrm{F},\phi}
=
\frac{r_\mathrm{F}}{\varepsilon}e^{i\phi}.
\end{equation}
Moving to the displaced \(c\)-mode frame gives
\begin{equation}
\begin{aligned}
D_c^\dagger(\zeta_{\mathrm{F},\phi})cD_c(\zeta_{\mathrm{F},\phi})
&=
c+\zeta_{\mathrm{F},\phi},
\\
D_c^\dagger(\zeta_{\mathrm{F},\phi})c^\dagger D_c(\zeta_{\mathrm{F},\phi})
&=
c^\dagger+\zeta_{\mathrm{F},\phi}^* .
\end{aligned}
\end{equation}

Indeed, using the expansion of the exponential, one has
\begin{equation}
\begin{aligned}
D^\dagger e^A D
&=
D^\dagger
\left(
\sum_{n=0}^{\infty}\frac{A^n}{n!}
\right)
D =
\sum_{n=0}^{\infty}
\frac{D^\dagger A^nD}{n!}.
\end{aligned}
\end{equation}
Since \(D^\dagger D=DD^\dagger=I\), it follows that
$D^\dagger A^nD
=D^\dagger (ADD^\dagger)^nD=
(D^\dagger A D)^n$.
Therefore
\begin{equation}
D^\dagger e^A D
=
\sum_{n=0}^{\infty}
\frac{(D^\dagger A D)^n}{n!}
=
e^{D^\dagger A D}.
\end{equation}
Therefore the exponent of \(U_\varepsilon\) becomes
\begin{align}
&-i\varepsilon
\left[
(c^\dagger+\zeta_{\mathrm{F},\phi}^*)ab
+
(c+\zeta_{\mathrm{F},\phi})a^\dagger b^\dagger
\right]
\nonumber\\
&\quad =
-i r_\mathrm{F}
\left(
e^{-i\phi}ab+e^{i\phi}a^\dagger b^\dagger
\right)
-i\varepsilon
\left(
c^\dagger ab+ca^\dagger b^\dagger
\right)
.
\end{align}
Thus, in the limit
\begin{equation}
\varepsilon\to0,
\qquad
\left|\frac{r_\mathrm{F}}{\varepsilon}\right|\to\infty,
\qquad
r_\mathrm{F}\ \text{fixed},
\end{equation}
the term $-i\varepsilon
\left(
c^\dagger ab+ca^\dagger b^\dagger
\right)$ vanishes
and  with \(\vartheta=\phi-\pi/2\), 
the limiting unitary is a phase-dependent two-mode squeezing unitary
\begin{equation}
U_{\mathrm{TMS},\phi}^{ab}(r_\mathrm{F})
=
\exp\left[
r_\mathrm{F}
\left(
e^{i\vartheta}a^\dagger b^\dagger
-
e^{-i\vartheta}ab
\right)
\right].
\end{equation}
It satisfies
\begin{equation}\label{D17}
\left(U_{\mathrm{TMS},\phi}^{ab}\right)^\dagger
a
U_{\mathrm{TMS},\phi}^{ab}
=
\cosh r_\mathrm{F}\,a
+
e^{i\vartheta}\sinh r_\mathrm{F}\,b^\dagger .
\end{equation}
The corresponding quadrature transformation is
\begin{equation}
\hat{\bm r}_{a,\mathrm{out}}
=
\cosh r_\mathrm{F}\,\hat{\bm r}_{a}
+
\sinh r_\mathrm{F}\,M(\vartheta)\hat{\bm r}_{b},
\end{equation}
where
\begin{equation}
M(\vartheta)=
\begin{pmatrix}
\cos\vartheta & \sin\vartheta\\
\sin\vartheta & -\cos\vartheta
\end{pmatrix}.
\end{equation}
Therefore, for an input covariance matrix \(V\) of the \(a\) mode and an
environmental covariance matrix \(V_b\) of the \(b\) mode, one has
\begin{equation} \label{D20}
V
\longmapsto
\cosh^2 r_\mathrm{F}\,V
+
\sinh^2 r_\mathrm{F}\,M(\vartheta)V_bM(\vartheta)^T .
\end{equation}
Since the \(b\)-mode environment is thermal, $V_b=(2\bar n+1)I$.
Moreover, \(M(\vartheta)M(\vartheta)^T=I\). Hence
$M(\vartheta)V_bM(\vartheta)^T
=
(2\bar n+1)I$.
Thus the phase factor does not affect the induced channel, and we finally obtain
\begin{equation}
V
\longmapsto
\cosh^2 r_\mathrm{F}\,V
+
\sinh^2 r_\mathrm{F}(2\bar n+1)I .
\end{equation}
According to~\eqref{C_state}, $\cosh^2 r_\mathrm{F}=\tau,
\
\sinh^2 r_\mathrm{F}=\tau-1$,
we obtain
\begin{equation} \label{eq:forward_amp_limit}
\mathcal{E}_\mathrm{F}^{(\varepsilon)}
\longrightarrow
\mathcal{E}_{A},
\qquad
\varepsilon\to0 .
\end{equation}
Therefore, under the classical-pump limit, the forward channel~\eqref{Forward_channel_qutrit} converges to the single-mode Gaussian amplification channel.

\textbf{Petz recovery map of the amplification channel.}
We now identify the Petz recovery map of the limiting amplification
channel~\eqref{Gaussian_amplification} by using the Gaussian Petz formula in
Eqs.~\eqref{eq:gaussian_petz_recovery}--\eqref{deltaP}.
Take the reference state to be a full-rank thermal state
$\alpha=\rho_{\mathrm{th}}(\bar n_\alpha),
\
V_\alpha=(2\bar n_\alpha+1)I$.
Hence the output reference state $\alpha'=\mathcal{E}_{A}(\alpha)$ is also thermal, with covariance
$V_{\alpha'}=(2\bar n_{\alpha'}+1)I$,
where
$2\bar n_{\alpha'}+1
=
\tau(2\bar n_\alpha+1)
+
(\tau-1)(2\bar n+1)$.
From~\eqref{XP}, we obtain
\begin{equation}\label{eq:XP_amp_appendix}
X_\mathrm{P}
=
\sqrt{\tau}
\frac{
\sqrt{(2\bar n_\alpha+1)^2-1}
}{
\sqrt{(2\bar n_{\alpha'}+1)^2-1}
}I
:=
\sqrt{\eta_\mathrm{P}}\,I,
\end{equation}
which has $0<\eta_\mathrm{P}<1$. Therefore, the Petz recovery map has the linear part of an attenuator.
Moreover, Eq.~\eqref{YP} gives
\begin{equation}\label{eq:YP_amp_appendix}
\begin{aligned}
Y_\mathrm{P}
&=
\left[
(2\bar n_\alpha+1)
-
\eta_\mathrm{P}(2\bar n_{\alpha'}+1)
\right]I .
\end{aligned}
\end{equation}
Since \(\vec d_\alpha=\vec d_{\alpha'}=0\), Eq.~\eqref{deltaP} gives
\begin{equation}
\vec\delta_\mathrm{P}=0 .
\end{equation}
Therefore, the Petz recovery map of the amplification channel with respect
to the thermal reference state \(\alpha\) is
\begin{equation}\label{eq:petz_amp_attenuator}
\widehat{\mathcal{E}}_{A,\alpha}:
\
V
\longmapsto
\eta_\mathrm{P}V+Y_\mathrm{P},
\
\vec d\longmapsto \sqrt{\eta_\mathrm{P}}\,\vec d .
\end{equation}
This is a Gaussian attenuator with transmissivity $\eta_{\mathrm P}$ and noise matrix $Y_{\mathrm P}$.

\textbf{Reverse channel.}
In the tabletop reverse channel, the same unitary is used in reverse, namely
\(U_\varepsilon^\dagger\). We choose the reverse environmental state as
\begin{equation}
\eta_\mathrm{B}^{(\varepsilon)}
=
\omega_{\mathrm{R},b}^{(\varepsilon)}
\otimes
\rho_\mathrm{th}^{(c)}(\bar n_\mathrm{R}),
\end{equation}
where
\begin{equation} \label{D28}
\omega_{\mathrm{R},b}^{(\varepsilon)}
:=
\int_0^{2\pi}
\frac{d\phi}{2\pi}
\left|
\frac{\theta}{\varepsilon}e^{i\phi}
\right\rangle_b
\left\langle
\frac{\theta}{\varepsilon}e^{i\phi}
\right|,
\
\theta:=\arccos\sqrt{\eta_\mathrm{P}} .
\end{equation}
Equivalently,
\begin{equation}
\omega_{\mathrm{R},b}^{(\varepsilon)}
=
\exp\left[-\frac{\theta^2}{\varepsilon^2}\right]
\sum_{n=0}^{\infty}
\frac{
\left(\theta^2/\varepsilon^2\right)^n
}{n!}
|n\rangle_b\langle n|.
\end{equation}

From Eqs.~\eqref{duiyi2} and \eqref{duiyi3}, we obtain
$[\eta_\mathrm{B}^{(\varepsilon)},H_b+H_c]=0$,
and
$[U_\varepsilon^\dagger,H_a+H_b+H_c]=0$.
Consequently, the reverse channel
\begin{equation}
\mathcal{E}_\mathrm{B}^{(\varepsilon)}(\rho_a')
=
\operatorname{Tr}_{bc}
\left[
U_\varepsilon^\dagger
\left(
\rho_a'\otimes\eta_\mathrm{B}^{(\varepsilon)}
\right)
U_\varepsilon
\right]
\end{equation}
is phase covariant.

For each fixed phase \(\phi\), moving to the displaced \(b\)-mode frame gives
\begin{align}
&i\varepsilon
\left[
c^\dagger a(b+\frac{\theta}{\varepsilon}e^{i\phi})
+
ca^\dagger(b^\dagger+\frac{\theta}{\varepsilon}e^{-i\phi})
\right]
\nonumber\\
&\quad =
i\varepsilon
\left(
c^\dagger ab+ca^\dagger b^\dagger
\right)
+
i\theta
\left(
e^{i\phi}c^\dagger a
+
e^{-i\phi}ca^\dagger
\right).
\end{align}
With \(\chi=\phi+\pi/2\), the leading term becomes
\begin{equation}
\theta
\left(
e^{i\chi}ac^\dagger
-
e^{-i\chi}a^\dagger c
\right).
\end{equation}
Therefore, in the limit
\begin{equation}
\varepsilon\to0,
\qquad
\left|\frac{\theta}{\varepsilon}\right|\to\infty,
\qquad
\theta\ \text{fixed},
\end{equation}
and with \(\chi=\phi+\pi/2\), the limiting unitary is a phase-dependent beam-splitter unitary
\begin{equation}
U_{\mathrm{BS},\phi}^{ac}(\theta)
=
\exp\left[
\theta
\left(
e^{i\chi}ac^\dagger
-
e^{-i\chi}a^\dagger c
\right)
\right].
\end{equation}
It satisfies
\begin{equation}
\left(U_{\mathrm{BS},\phi}^{ac}\right)^\dagger
a
U_{\mathrm{BS},\phi}^{ac}
=
\cos\theta\,a
-
e^{-i\chi}\sin\theta\,c.
\end{equation}
Following the same procedure as in Eqs.~\eqref{D17}--\eqref{D20}, since the \(c\) mode is thermal, the phase factor does not affect the induced
channel on the \(a\) mode. Substituting \(\theta=\arccos\sqrt{\eta_{\mathrm{P}}}\) according to Eq.~\eqref{D28}, we obtain
\begin{equation}
V
\longmapsto
\eta_\mathrm{P}V
+
(1-\eta_\mathrm{P})(2\bar n_\mathrm{R}+1)I.
\end{equation}
To make the limiting reverse channel coincide
with the Petz recovery map in Eq.~\eqref{eq:petz_amp_attenuator}, we choose
the thermal photon number \(\bar n_\mathrm{R}\) of the \(c\)-mode environment
such that
\begin{equation}
Y_\mathrm{P}=(1-\eta_\mathrm{P})(2\bar n_\mathrm{R}+1)I.
\end{equation}
Equivalently,
\begin{equation}
2\bar n_\mathrm{R}+1
=
\frac{
(2\bar n_\alpha+1)
-
\eta_\mathrm{P}(2\bar n_{\alpha'}+1)
}{
1-\eta_\mathrm{P}
}.
\end{equation}

One can prove in the above equation $2\bar n_\mathrm{R}+1\ge1$, so this definition is well defined and gives a physical thermal photon number.  

Therefore, we obtain
\begin{equation}\label{eq:backward_Petz_limit}
\mathcal{E}_\mathrm{B}^{(\varepsilon)}
\longrightarrow
\widehat{\mathcal{E}}_{A,\alpha},
\qquad
\varepsilon\to0 .
\end{equation}
Combining Eqs.~\eqref{eq:forward_amp_limit} and~\eqref{eq:backward_Petz_limit}, the same three-mode non-Gaussian unitary architecture realizes
$\mathcal{E}_\mathrm{F}^{(\varepsilon)}
\longrightarrow
\mathcal{E}_{A},
\
\mathcal{E}_\mathrm{B}^{(\varepsilon)}
\longrightarrow
\widehat{\mathcal{E}}_{A,\alpha}$.
Hence, the single-mode Gaussian amplification channel admits a non-Gaussian tabletop time-reversible dilation.

\textbf{Further discussion.}
It would be interesting future work to investigate whether TTR
conditions can be established for multimode Gaussian channels, and whether
every continuous-variable channel admits non-Gaussian dilations.
The construction above suggests a possible intuition in this direction. 
In continuous-variable systems, a general channel may involve
infinitely many Kraus operators, and each Kraus operator acts on an
infinite-dimensional Hilbert space~\cite{ivan2011operator}. As a result, the finite-dimensional
unitary-completion argument used in Lemma~\ref{lemma1} can no longer be applied
directly. In particular, such a construction would require prescribing infinitely many
rows and columns of an infinite-dimensional operator and then completing them
to a unitary, which constitutes a substantial constructive obstruction. We thus leave this as an open question.



%

\end{document}